\documentclass[preprint,12pt]{elsarticle}

\usepackage[T1]{fontenc}
\usepackage[utf8]{inputenc}
\usepackage{lmodern}
\usepackage{microtype}

\usepackage{amsmath,amssymb,amsthm,mathtools}
\usepackage{mathrsfs}
\usepackage{stmaryrd}
\usepackage{bm}
\usepackage{mathpartir}

\usepackage{booktabs}
\usepackage{enumitem}
\usepackage{tikz}
\usepackage{tikz-cd}
\usepackage{listings}

\usepackage[hidelinks]{hyperref}
\usepackage[nameinlink,noabbrev]{cleveref}

\journal{Information and Computation}

\newtheorem{theorem}{Theorem}[section]
\newtheorem{proposition}[theorem]{Proposition}
\newtheorem{lemma}[theorem]{Lemma}
\newtheorem{corollary}[theorem]{Corollary}
\newtheorem{definition}[theorem]{Definition}
\newtheorem{assumption}[theorem]{Assumption}

\newcommand{\N}{\mathbb{N}}
\newcommand{\R}{\mathbb{R}}
\newcommand{\C}{\mathbb{C}}
\newcommand{\id}{\operatorname{id}}
\newcommand{\fix}{\operatorname{fix}}
\newcommand{\supp}{\operatorname{supp}}

\newcommand{\Tr}{\operatorname{Tr}}
\newcommand{\DetF}{\operatorname{det}_{\mathrm F}}
\newcommand{\Spec}{\operatorname{Spec}}

\newcommand{\DCPO}{\mathbf{DCPO}}
\newcommand{\Wstar}{W^{*}}
\newcommand{\QO}{\mathsf{Q}}
\newcommand{\Exec}{\mathsf{Exec}}

\newcommand{\sem}[1]{\llbracket #1\rrbracket}
\newcommand{\kleisli}[1]{#1^{\sharp}}
\newcommand{\bottom}{\bot}

\newcommand{\Sone}{\mathfrak{S}_{1}}
\newcommand{\Stwo}{\mathfrak{S}_{2}}
\newcommand{\CR}{\operatorname{CR}}

\begin{document}

\begin{frontmatter}

\title{Graph-Series Semantics and Abel Regularization for Recursive Hybrid Quantum Programs}

\author[angers,lepage,jeannedarc]{Jean-Pierre Magnot\corref{cor1}}
\ead{magnot@math.cnrs.fr}
\cortext[cor1]{Corresponding author.}

\address[angers]{Univ Angers, CNRS, LAREMA, SFR MATHSTIC, F-49000 Angers, France}
\address[lepage]{Lepage Research Institute, 17 novembra 1, 081 16 Pre\v{s}ov, Slovakia}
\address[jeannedarc]{Lyc\'ee Jeanne d'Arc, Avenue de Grande Bretagne, 63000 Clermont-Ferrand, France}

\begin{abstract}
	We introduce a graded graph-series semantics for recursive hybrid quantum
	programs interpreted in the quantum orchestra monad. Finite terminating
	executions are represented by directed paths whose edges carry normal
	completely positive subunital maps and whose terminal vertices carry classical
	results. Path concatenation defines a graded execution category, while
	continuation grafting models outcome-dependent sequential composition. We
	construct a semantic evaluation from admissible execution-graph series to
	quantum orchestras and prove that it is compatible with both channel
	composition and Kleisli composition.
	
	For finitary recursive programs, the truncation of the execution series at
	degree $n$ is shown to coincide with the $n$-th Kleene approximant of the
	associated Scott-continuous recursion functional. Consequently, evaluation of
	the complete graph series recovers the ordinary least-fixed-point denotation.
	Weighting a graph of degree $n$ by $q^n$, with $0<q<1$, yields an
	Abel-regularised semantics whose Scott limit as $q\to 1^{-}$ is the
	unregularised recursive denotation. Equivalently, the parametrisation
	$q=e^{-t}$ exponentially suppresses long executions and reconstructs the
	denotation as $t\to 0^{+}$.
	
	In a supplementary linear feedback sector, repeated recursion is represented
	by the execution resolvent $(I-qST)^{-1}$. We identify $I-qST$ with an
	algebraic cross-ratio of graph subspaces. Under Hilbert--Schmidt assumptions,
	the associated return operator is trace class and defines the Fredholm
	feedback determinant
	\(
	\operatorname{det}_{F}(I-qST),
	\)
	whose zeros detect singular feedback configurations and whose logarithmic
	expansion records closed loop traversals.
\end{abstract}


\end{frontmatter}

% -----------------------------------------------------------------------------
% Provisional structure
% -----------------------------------------------------------------------------

\section{Introduction}
\label{sec:introduction}

Quantum programming languages must combine two qualitatively different forms
of computation. Quantum data evolve through completely positive maps, while
classical control selects subsequent operations, stores measurement outcomes,
and determines whether a computation terminates or recurses. Early categorical
and type-theoretic models established the basic interaction between quantum
operations and classical control \cite{SelingerValiron2006,Selinger2007}, and
scalable language design subsequently made this interaction explicit at the
level of circuit generation and higher-order programming
\cite{GreenLumsdaineRossSelingerValiron2013}. More recent semantic accounts
have treated quantum computation through algebraic effects
\cite{Staton2015}, operator algebras \cite{Cho2016}, categorical completions
of quantum channels \cite{HuotStaton2019}, information effects
\cite{HeunenKaarsgaard2022}, and domain-theoretic models supporting
continuous or variational parameters
\cite{JiaKornellLindenhoviusMisloveZamdzhiev2022}. The operator-algebraic
background for the present work is the standard theory of $C^*$- and
$W^*$-algebras \cite{Sakai1971}.

A central technical issue is the semantic treatment of measurement-dependent
recursion. Quantum measurements are represented by instruments, which jointly
encode the classical outcome and the corresponding quantum state
transformation \cite{DaviesLewis1970}. Their compositional structure has
recently been clarified from complementary perspectives. Fritz constructs a
quantum instrument monad \cite{Fritz2026}, while Booth, Leichtle, Rice, and
Worrall study the composition of quantum instruments and the interaction
between outcome-dependent control and quantum channels
\cite{BoothLeichtleRiceWorrall2026}. Building on this analysis, Rice,
Leichtle, Worrall, and Booth introduce quantum orchestras, a
continuation-based semantics for recursive hybrid programs
\cite{RiceLeichtleWorrallBooth2026}. Quantum orchestras form a pointed
directed-complete semantic domain in which recursive programs are interpreted
as least fixed points of Scott-continuous maps. They provide the semantic
starting point of the present paper.
The use of monads to represent computational effects originates in Moggi's
seminal formulation of notions of computation \cite{Moggi1991}. Parameterised
monads refine this framework by allowing the effectful interface to change
across a computation \cite{Atkey2009}, while related indexed constructions
model environments and state-dependent interfaces
\cite{LevyPowerThielecke2003}. Effect systems can be interpreted through
parametric effect monads \cite{Katsumata2014}, and graded monads provide a
flexible means of recording quantitative or structural information about
computation \cite{KatsumataMcDermottUustaluWu2022}. The graph degree used in
this paper is not introduced as an additional effect system. Instead, it is a
combinatorial refinement of an already defined quantum-orchestra denotation:
it records the length of finite execution histories before those histories
are summed in the semantic domain.

The treatment of recursion relies on the classical order-theoretic theory of
fixed points. Tarski's lattice-theoretic theorem provides the general fixed
point principle for monotone maps on complete lattices \cite{Tarski1955},
while domain theory identifies Scott continuity and directed completeness as
the appropriate structures for denotational approximation
\cite{AbramskyJung1994,StoltenbergHansenLindstromGriffor1994}. Iteration
theories axiomatise the equational properties of recursive and iterative
processes \cite{BloomEsik1993}. In probabilistic semantics, the Giry monad
\cite{Giry1982}, probabilistic powerdomains \cite{JonesPlotkin1989}, and mixed
powerdomains combining probability and nondeterminism
\cite{KeimelPlotkin2017} show how infinite behaviour can be recovered from
ordered families of finite approximations. Quantum orchestras extend this
order-theoretic pattern to measurement-dependent quantum computations. The
present paper asks whether the individual finite histories underlying the
Kleene chain can themselves be organised into a compositional formal series.

Our answer is based on execution graphs. A terminating run is represented by
a finite directed path whose edges carry normal completely positive subunital
maps and whose terminal vertex carries a classical result. Concatenation of
paths is compatible with composition of quantum channels, and the path length
defines an additive grading. Formal sums of such paths therefore form a
complete graded algebra of the type considered in the theory of generalised
formal series indexed by graded categories
\cite{Magnot2025FormalSeries}. The idea of refining a nonlinear or recursive
object by a formal expansion is also reminiscent of Taylor expansions of
lambda terms \cite{EhrhardRegnier2008} and of their coherent organisation by
monadic and bimonadic structures \cite{EhrhardWalch2025}. The construction
developed here is different in purpose and coefficient semantics: its
coefficients are quantum execution histories, and evaluation sends them to
quantum orchestras. Nevertheless, these works motivate the systematic
separation between a formal expansion and the semantic operation that resums
it.

The first main result is a compositional graph-series semantics. Finite graph
polynomials are evaluated by summing the Dirac orchestras associated with
their terminating paths. Path concatenation evaluates to channel composition,
while grafting a continuation graph after a terminating path evaluates to
Kleisli composition in the quantum orchestra monad. For a finitary recursive
program, the degree-$n$ truncation of its execution series is exactly the
$n$-th Kleene approximant. Consequently, semantic evaluation of the complete
graph series recovers the ordinary least-fixed-point denotation. The graph
series is therefore a conservative refinement: it retains individual
histories and their degrees, but forgetting this extra information gives the
original quantum-orchestra semantics.

The grading also provides a canonical regularisation. For $0<q<1$, a graph
of degree $n$ is multiplied by $q^n$. The resulting positive family is
monotone in $q$, and its Scott supremum as $q\to 1^{-}$ is the unregularised
recursive denotation. Equivalently, with $q=e^{-t}$, the parameter $t>0$
exponentially suppresses long executions and the limit $t\to 0^{+}$ removes
the cutoff. The relation between exact-depth contributions and cumulative
Kleene approximants yields the Abel identity
\[
\sum_{n=0}^{\infty} q^n D_n
=
(1-q)\sum_{n=0}^{\infty} q^n A_n,
\]
where $D_n$ is the contribution of histories of exact degree $n$ and $A_n$ is
the $n$-th cumulative approximant. This is an order-theoretic form of Abel
summation. Its classical analytic foundations belong to the theory of
divergent series \cite{Hardy1949} and to Tauberian theory
\cite{Korevaar2004}, but the basic reconstruction theorem proved here does
not require norm convergence or a Tauberian converse.

A second, supplementary part of the paper concerns linear feedback. Traced
monoidal categories provide an abstract language for feedback
\cite{JoyalStreetVerity1996}, and their relation with cyclic sharing and
recursion was developed by Hasegawa \cite{Hasegawa1997}. Guarded traced
categories distinguish well-founded feedback from unrestricted cyclicity
\cite{GoncharovSchroder2018}. Geometry-of-interaction models interpret
computation through the circulation of information along feedback paths
\cite{HaghverdiScott2006}, including higher-order quantum computation
\cite{HasuoHoshino2017}. In the linear sector considered here, an open
computation is represented by $T:C\to R$ and the return interface by
$S:R\to C$. Repeated feedback is then governed by the formal or analytic
resolvent
\[
(I_C-qST)^{-1}
=
\sum_{n=0}^{\infty} q^n (ST)^n.
\]
This series is the direct linear image of repeated loop traversals. The same
operator is an algebraic cross-ratio of two graph decompositions of
$C\oplus R$, in the sense suggested by the general theory of noncommutative
cross-ratio mappings and infinite-dimensional splittings
\cite{Magnot2024CrossRatio}. No smooth Grassmannian structure is needed for
the algebraic identity used in this paper.

When $C$ and $R$ are Hilbert spaces and the return operator $ST$ is trace
class, the feedback operator has the Fredholm determinant
\[
\Delta_{T,S}(q)
=
\DetF(I_C-qST).
\]
The relevant trace-ideal and determinant theory is classical
\cite{GohbergKrein1969,Simon2005}. The zeros of this determinant detect
singular feedback configurations, while
\[
\log\Delta_{T,S}(q)
=
-\sum_{m=1}^{\infty}
\frac{q^m}{m}
\Tr\bigl((ST)^m\bigr)
\]
organises traces of closed loop traversals. This determinant is a secondary
invariant requiring additional Hilbert--Schmidt or trace-class assumptions;
it is not part of the general dcpo semantics. In particular, we do not claim
that it is a tau function or that it satisfies Pl"ucker or Hirota
relations.

Repeat-until-success circuits provide a basic test case for the interaction
between recursion, measurement, and feedback. Such circuits implement a
desired quantum operation by repeating a probabilistic subroutine until a
success outcome is obtained \cite{PaetznickSvore2014}. In this example, the
execution series is geometric, its Abel sum is explicit, and the graph-series
denotation agrees with the least fixed point. More elaborate examples show
that the dcpo limit may remain meaningful even when norm convergence is not
the primary notion, whereas Fredholm determinants require separate analytic
control.

The contributions of the paper can be summarised as follows.
\begin{enumerate}
	\item We define a graded execution category and a complete algebra of formal
	execution-graph series for finitary hybrid quantum programs.
	
	\item We construct a semantic evaluation into quantum orchestras and prove
	that graph concatenation and continuation grafting are compatible with
	channel composition and Kleisli composition.
	
	\item We prove that finite degree truncations coincide with Kleene
	approximants and that the complete graph-series denotation equals the original
	least-fixed-point denotation.
	
	\item We introduce the degree regularisation $q^{|\Gamma|}$ and prove an Abel
	reconstruction theorem in the Scott order, together with norm estimates under
	additional geometric convergence assumptions.
	
	\item In a linear feedback sector, we identify the execution resolvent with
	the inverse of an algebraic cross-ratio and construct a Fredholm determinant
	under trace-class hypotheses.
\end{enumerate}

The paper is organised as follows. Section~\ref{sec:orchestras} recalls
quantum orchestras, instruments, and recursive denotational semantics.
Section~\ref{sec:graphs} introduces finite execution graphs and their graded
formal series. Section~\ref{sec:graph-semantics} defines compositional
semantic evaluation. Section~\ref{sec:abel} proves the finite-unfolding
correspondence and the Abel reconstruction theorem.
Section~\ref{sec:feedback} develops the algebraic feedback resolvent and its
cross-ratio interpretation. Section~\ref{sec:fredholm} studies the
trace-class sector and the Fredholm invariant. Section~\ref{sec:examples}
contains explicit examples, and Section~\ref{sec:discussion} discusses the
scope, limitations, and possible extensions of the construction.

\section{Quantum orchestras and recursive denotational semantics}
\label{sec:orchestras}

This section recalls the semantic framework on which the constructions of the
following sections are based. We use the quantum orchestra monad introduced in
\cite{RiceLeichtleWorrallBooth2026} to model hybrid computations in which a
classical program interacts with an evolving quantum state, receives classical
outcomes from measurements, and uses these outcomes to determine both the
subsequent quantum operations and the termination behaviour. The construction
combines the operator-algebraic formalism of quantum instruments with the
order-theoretic treatment of recursion by directed-complete partial orders.
Throughout the paper, quantum transformations are written in the Heisenberg
picture.

\subsection{Quantum channels and domain-theoretic structure}
\label{subsec:channels-domains}

Let $\mathcal A$ and $\mathcal B$ be von Neumann algebras. We write
\[
\Wstar(\mathcal B,\mathcal A)
\]
for the set of normal completely positive subunital maps
\[
\Phi:\mathcal B\longrightarrow\mathcal A.
\]
Thus, $\Phi$ is ultraweakly continuous, every matrix amplification of $\Phi$
is positive, and
\[
\Phi(1_{\mathcal B})\leq 1_{\mathcal A}.
\]
Composition of such maps is again normal, completely positive, and subunital.
The resulting category of von Neumann algebras and normal completely positive
subunital maps is symmetric monoidal for the spatial tensor product. In the
Heisenberg convention used here, a morphism $\Phi:\mathcal B\to\mathcal A$
represents a quantum process from the system described by $\mathcal A$ to the
system described by $\mathcal B$; observables are pulled backwards along
$\Phi$. We refer to \cite{Sakai1971,Cho2016} for the operator-algebraic
background.

The hom-set $\Wstar(\mathcal B,\mathcal A)$ is equipped with the completely
positive order
\[
\Phi\leq\Psi
\quad\Longleftrightarrow\quad
\Psi-\Phi\text{ is completely positive}.
\]
Its least element is the zero map. A fundamental fact used in the quantum
orchestra construction is that $\Wstar(\mathcal B,\mathcal A)$ is a pointed
directed-complete partial order and that composition is Scott-continuous in
each variable \cite{Cho2016,RiceLeichtleWorrallBooth2026}.

We recall the domain-theoretic terminology. A subset $D$ of a partially
ordered set $X$ is directed if every finite subset of $D$ has an upper bound
in $D$. A directed-complete partial order, abbreviated as dcpo, is a partially
ordered set in which every directed subset has a supremum. A dcpo is pointed
if it has a least element, denoted by $\bottom$. A map $f:X\to Y$ between
dcpos is Scott-continuous if it preserves directed suprema. We denote by
$\DCPO$ the category of dcpos and Scott-continuous maps.

The order-theoretic interpretation of recursion is supplied by the Kleene
fixed-point theorem
\cite{AbramskyJung1994,StoltenbergHansenLindstromGriffor1994}.
If $X$ is a pointed dcpo and $F:X\to X$ is Scott-continuous, then $F$ has a
least fixed point given by
\[
\fix(F)
=
\sup_{n\in\N}F^{n}(\bottom).
\]
The sequence
\[
\bottom\leq F(\bottom)\leq F^{2}(\bottom)\leq\cdots
\]
is called the Kleene chain of $F$. Its elements will later be related to
finite execution graphs and to the coefficients of the graph-series
semantics.

\subsection{Finite quantum instruments}
\label{subsec:finite-instruments}

A quantum measurement produces both a classical outcome and a transformed
quantum state. Quantum instruments provide the standard operator-algebraic
representation of this joint behaviour \cite{DaviesLewis1970}.

\begin{definition}
	\label{def:finite-instrument}
	Let $X$ be a set and let $\mathcal A$ and $\mathcal B$ be von Neumann
	algebras. A finite quantum instrument from $\mathcal A$ to $\mathcal B$ with
	classical outcomes in $X$ is a family
	\[
	\Phi=(\Phi_x)_{x\in X},
	\qquad
	\Phi_x\in\Wstar(\mathcal B,\mathcal A),
	\]
	such that
	\[
	\supp(\Phi)
	=
	\{x\in X\mid \Phi_x\neq 0\}
	\]
	is finite and the map
	\[
	\sum_{x\in X}\Phi_x
	\]
	is subunital.
\end{definition}

The component $\Phi_x$ describes the quantum evolution conditioned on the
classical outcome $x$. If $\rho$ is a normal state on $\mathcal A$, then
\[
\rho\bigl(\Phi_x(1_{\mathcal B})\bigr)
\]
is the probability of observing $x$, while
\[
\sum_{x\in X}\rho\bigl(\Phi_x(b)\bigr)
\]
is the expectation of an observable $b\in\mathcal B$ when the classical
outcome is discarded.

For $x\in X$ and $\Phi\in\Wstar(\mathcal B,\mathcal A)$, the associated
Dirac instrument is denoted by $\delta(x,\Phi)$ and is defined by
\[
\delta(x,\Phi)_y
=
\begin{cases}
	\Phi, & y=x,\\
	0, & y\neq x.
\end{cases}
\]
Every finite instrument is a finite sum of Dirac instruments. This notation
also makes sequential composition transparent. If an instrument first returns
$x\in X$ and the continuation selected by this outcome is an instrument
$f(x)$, then the appropriate composition is the Kleisli composition of the
first instrument with the outcome-dependent family $f$, rather than
composition with a single fixed channel. This is the mechanism by which
measurement-dependent classical control is represented.

Finite instruments are sufficient for finite classical outcome spaces and
bounded computations. They are not, however, closed under general recursion.
For example, a program returning the number of iterations performed before
termination may have nonzero components at every natural number. Moreover, a
pointwise order on outcome-indexed families does not provide a satisfactory
continuous unit when the outcome space itself carries a nontrivial dcpo
structure. These difficulties motivate the continuation-based completion by
quantum orchestras.

\subsection{The quantum orchestra monad}
\label{subsec:quantum-orchestra-monad}

The quantum orchestra construction is parameterised by an input algebra, an
output algebra, and a classical result dcpo. The algebra parameters encode the
quantum interface, while the dcpo parameter carries the classical return
values and their approximation order. The parameterised form is in the spirit
of Atkey's parameterised monads \cite{Atkey2009}; the underlying
interpretation of computational effects follows the monadic approach of
\cite{Moggi1991}.

\begin{definition}
	\label{def:quantum-orchestra}
	Let $\mathcal A$ and $\mathcal B$ be von Neumann algebras and let $X$ be a
	dcpo. A quantum orchestra
	\[
	\xi\in\QO(\mathcal A,\mathcal B,X)
	\]
	is a family of maps
	\[
	\xi_{\mathcal K}:
	\DCPO\bigl(X,\Wstar(\mathcal K,\mathcal B)\bigr)
	\longrightarrow
	\Wstar(\mathcal K,\mathcal A),
	\]
	indexed by von Neumann algebras $\mathcal K$, satisfying the following
	conditions.
	
	\begin{enumerate}[label=\textup{(\roman*)}]
		\item For every $\mathcal K$, the map
		\[
		k\longmapsto\xi_{\mathcal K}(k)
		\]
		is Scott-continuous.
		
		\item For all $\theta,\rho\in\R_{+}$ with $\theta+\rho\leq 1$ and all
		\[
		k,k'\in\DCPO\bigl(X,\Wstar(\mathcal K,\mathcal B)\bigr),
		\]
		one has
		\[
		\xi_{\mathcal K}(\theta k+\rho k')
		=
		\theta\xi_{\mathcal K}(k)+\rho\xi_{\mathcal K}(k').
		\]
		
		\item For every
		\[
		\chi\in\Wstar(\mathcal K',\mathcal K)
		\]
		and every
		\[
		k\in\DCPO\bigl(X,\Wstar(\mathcal K,\mathcal B)\bigr),
		\]
		one has
		\[
		\xi_{\mathcal K'}
		\bigl(x\mapsto k(x)\circ\chi\bigr)
		=
		\xi_{\mathcal K}(k)\circ\chi.
		\]
	\end{enumerate}
\end{definition}

When no confusion can arise, we write $\xi_k$ instead of
$\xi_{\mathcal K}(k)$. The continuation
\[
k:X\longrightarrow\Wstar(\mathcal K,\mathcal B)
\]
assigns to every classical return value a subsequent quantum computation. The
orchestra evaluates the whole continuation and returns a channel in
$\Wstar(\mathcal K,\mathcal A)$. The compositionality axiom states that a
channel attached after the continuation can be moved outside the orchestra.
The subconvexity axiom is the noncommutative analogue of the linearity
satisfied by integration against a subprobability valuation.

The order on $\QO(\mathcal A,\mathcal B,X)$ is defined pointwise:
\[
\xi\leq\zeta
\quad\Longleftrightarrow\quad
\xi_k\leq\zeta_k
\quad\text{for every admissible continuation }k.
\]
Directed suprema are also computed pointwise:
\[
\left(\sup_{i\in I}\xi_i\right)_k
=
\sup_{i\in I}(\xi_i)_k.
\]
Consequently, $\QO(\mathcal A,\mathcal B,X)$ is a pointed dcpo. Its least
element is the orchestra
\[
\bottom_k=0
\]
for every continuation $k$.

For $x\in X$ and $\Phi\in\Wstar(\mathcal B,\mathcal A)$, the Dirac quantum
orchestra is defined by
\[
\delta(x,\Phi)_k
=
\Phi\circ k(x).
\]
In particular, the unit of the monad is
\[
\eta_X^{\mathcal A}(x)
=
\delta(x,\id_{\mathcal A}),
\]
or equivalently
\[
\eta_X^{\mathcal A}(x)_k
=
k(x).
\]

Let
\[
f:X\longrightarrow\QO(\mathcal B,\mathcal C,Y)
\]
be Scott-continuous. Its Kleisli extension is the Scott-continuous map
\[
\kleisli{f}:
\QO(\mathcal A,\mathcal B,X)
\longrightarrow
\QO(\mathcal A,\mathcal C,Y)
\]
defined by
\[
\kleisli{f}(\xi)_k
=
\xi_{x\mapsto f(x)_k},
\]
where
\[
k\in\DCPO\bigl(Y,\Wstar(\mathcal K,\mathcal C)\bigr).
\]
For a Dirac orchestra this reduces to
\[
\kleisli{f}\bigl(\delta(x,\Phi)\bigr)_k
=
\Phi\circ f(x)_k.
\]
Thus, the first computation produces $x$ and the channel $\Phi$, after which
the continuation $f(x)$ selected by the classical outcome is executed. The
unit and Kleisli extension satisfy the monad laws, and the construction is a
strong parameterised monad on $\DCPO$
\cite{RiceLeichtleWorrallBooth2026}.

Every finite quantum instrument
\[
(\Phi_x)_{x\in X}
\]
embeds into the orchestra monad through
\[
\iota(\Phi)
=
\sum_{x\in\supp(\Phi)}\delta(x,\Phi_x).
\]
Its action on a continuation is
\[
\iota(\Phi)_k
=
\sum_{x\in\supp(\Phi)}\Phi_x\circ k(x).
\]
Hence, on finite classical outcome spaces, the continuation-based
construction recovers the usual composition of quantum instruments.

\subsection{Recursive terms and Kleene approximants}
\label{subsec:recursive-terms}

A recursive program is interpreted by a Scott-continuous endomap on an
appropriate orchestra dcpo. For simplicity, consider a recursive term whose
quantum input and output interfaces are both represented by the same von
Neumann algebra $\mathcal A$, and whose classical result type is interpreted
by a dcpo $X$. Its recursion functional has the form
\[
F:\QO(\mathcal A,\mathcal A,X)
\longrightarrow
\QO(\mathcal A,\mathcal A,X).
\]
The denotation of the recursive term is the least fixed point
\[
\sem{\operatorname{fix}F}
=
\fix(F)
=
\sup_{n\in\N}F^n(\bottom).
\]
We write
\[
A_n(F)=F^n(\bottom)
\]
for the $n$-th Kleene approximant. The sequence
$(A_n(F))_{n\in\N}$ is increasing and represents the behaviour obtained by
unfolding the recursion at most $n$ times. This interpretation of finite
unfoldings is the point of departure for the graph expansion developed below:
execution graphs of degree $n$ will encode the contributions appearing at the
$n$-th stage of the Kleene construction.

The restriction to identical input and output quantum interfaces is natural
for a loop. More generally, recursive allocation may be handled by choosing a
sufficiently large, possibly infinite-dimensional, von Neumann algebra, or by
a refined parameterised type system. These issues are independent of the
graph-series construction considered in this paper, and we shall work with a
fixed quantum interface whenever recursion is involved.

\subsection{A repeat-until-success computation}
\label{subsec:rus}

We conclude with the standard repeat-until-success example, which will also
serve as a test case for the graph-series and Abel constructions. Let
$\mathcal A=\mathcal B(H)$ and let $U:H\to H$ be unitary. We denote by
\[
\mathcal U:\mathcal A\longrightarrow\mathcal A,
\qquad
\mathcal U(M)=U^{*}MU,
\]
the corresponding Heisenberg channel. Suppose that a computation $C$ returns
a Boolean outcome and has denotation
\[
\sem{C}
=
\delta(\mathsf{true},p\mathcal U)
+
\delta(\mathsf{false},(1-p)\id_{\mathcal A}),
\qquad
0<p<1.
\]
The outcome $\mathsf{true}$ indicates success and applies $\mathcal U$, while
$\mathsf{false}$ leaves the quantum state unchanged and requests another
iteration. The corresponding recursion functional is
\[
F(\xi)
=
\delta((),p\mathcal U)+(1-p)\xi,
\qquad
\xi\in\QO(\mathcal A,\mathcal A,1).
\]
Its Kleene approximants are
\[
F^n(\bottom)
=
\delta\bigl((),\bigl(1-(1-p)^n\bigr)\mathcal U\bigr),
\qquad
n\in\N.
\]
Indeed, the coefficient $1-(1-p)^n$ is the probability that at least one of
the first $n$ trials succeeds. Since
\[
\sup_{n\in\N}\bigl(1-(1-p)^n\bigr)=1,
\]
Scott-continuity gives
\[
\fix(F)
=
\sup_{n\in\N}F^n(\bottom)
=
\delta((),\mathcal U).
\]
Thus, although the number of iterations is unbounded, the recursive program
denotes the desired unitary channel. The finite approximants retain the
complete unfolding information, whereas the least fixed point records only
the resummed behaviour. The purpose of the next sections is to organise the
finite unfoldings into a graded graph series and to compare its Abel limit
with the dcpo supremum above.

\section{Execution graphs and graded series}
\label{sec:graphs}

The least-fixed-point semantics recalled in Section~\ref{sec:orchestras}
describes a recursive computation as the supremum of its finite unfoldings.
In order to retain the combinatorial information that is lost after taking
this supremum, we associate finite execution graphs with the individual
histories of a hybrid program. Branching behaviour is represented by a family
of such graphs rather than by a single branching tree. Every graph then
carries an unambiguous quantum transformation, and its length is additive
under concatenation. This provides the grading required for the formal-series
construction.

We first work with a fixed quantum store represented by a von Neumann algebra
$\mathcal A$. This includes the recursive examples considered in this paper.
A version with varying quantum interfaces can be obtained by assigning a von
Neumann algebra to each control location and requiring the channels attached
to consecutive edges to have compatible domains and codomains.

\subsection{Finitary hybrid control systems}
\label{subsec:control-systems}

\begin{definition}
	\label{def:control-system}
	A finitary hybrid control system over $\mathcal A$ is a tuple
	\[
	\Sigma=(L,L_{\mathrm{term}},X,r,\mathcal C)
	\]
	with the following data.
	\begin{enumerate}
		\item $L$ is a set of control locations, and
		$L_{\mathrm{term}}\subseteq L$ is the set of terminal locations.
		
		\item $X$ is the classical result set, and
		\[
		r:L_{\mathrm{term}}\longrightarrow X
		\]
		assigns a result to every terminal location.
		
		\item $\mathcal C$ is a set of commands. Every command $c\in\mathcal C$
		has a source location $s(c)\in L\setminus L_{\mathrm{term}}$, a finite
		nonempty outcome set $O_c$, a target map
		\[
		\tau_c:O_c\longrightarrow L,
		\]
		and a family of channels
		\[
		\Phi_{c,o}\in\Wstar(\mathcal A,\mathcal A),
		\qquad o\in O_c,
		\]
		such that
		\[
		\sum_{o\in O_c}\Phi_{c,o}
		\]
		is subunital.
	\end{enumerate}
\end{definition}

The family $(\Phi_{c,o})_{o\in O_c}$ is a finite quantum instrument. The
classical outcome $o$ determines the next control location $\tau_c(o)$.
Unitary commands, measurements, classically controlled gates, and failing
operations are all covered by this definition. Purely classical control can
be represented by channels equal to either $\id_{\mathcal A}$ or the zero
map.

A control system describes the elementary transitions available to a program.
The program syntax determines which command is selected at each reachable
nonterminal execution prefix.

\subsection{Finite execution graphs}
\label{subsec:finite-execution-graphs}

\begin{definition}
	\label{def:execution-graph}
	A finite execution graph in $\Sigma$ is a finite directed path
	\[
	\Gamma=
	\bigl(
	\ell_0\xrightarrow{(c_1,o_1)}\ell_1
	\xrightarrow{(c_2,o_2)}\cdots
	\xrightarrow{(c_n,o_n)}\ell_n
	\bigr),
	\]
	where, for every $j\in\{1,\ldots,n\}$,
	\[
	s(c_j)=\ell_{j-1},
	\qquad
	\tau_{c_j}(o_j)=\ell_j.
	\]
	Its source, target, and degree are respectively
	\[
	s(\Gamma)=\ell_0,
	\qquad
	t(\Gamma)=\ell_n,
	\qquad
	|\Gamma|=n.
	\]
	For every $\ell\in L$, the empty graph at $\ell$ is denoted by
	$\mathbf 1_{\ell}$ and has degree zero.
\end{definition}

The quantum weight of an execution graph is the channel obtained by composing
the channels encountered along the path in execution order:
\[
\Phi_{\Gamma}
=
\Phi_{c_1,o_1}\circ\Phi_{c_2,o_2}\circ\cdots\circ\Phi_{c_n,o_n}.
\]
For the empty graph, we set
\[
\Phi_{\mathbf 1_{\ell}}=\id_{\mathcal A}.
\]
The order of composition follows the Heisenberg convention: the first command
pulls back the observable produced by the continuation of the path.

If $t(\Gamma_1)=s(\Gamma_2)$, their concatenation is denoted by
\[
\Gamma_1\star\Gamma_2.
\]
It is obtained by first following $\Gamma_1$ and then $\Gamma_2$. One has
\[
|\Gamma_1\star\Gamma_2|
=
|\Gamma_1|+|\Gamma_2|
\]
and
\[
\Phi_{\Gamma_1\star\Gamma_2}
=
\Phi_{\Gamma_1}\circ\Phi_{\Gamma_2}.
\]

\begin{proposition}
	\label{prop:execution-category}
	The control locations form the objects of a category $\Exec(\Sigma)$ whose
	morphisms are finite execution graphs and whose composition is concatenation.
	The degree map
	\[
	|\cdot|:\Exec(\Sigma)\longrightarrow\N
	\]
	is additive under composition. Moreover, the assignment
	\[
	\Gamma\longmapsto\Phi_{\Gamma}
	\]
	preserves identities and composition.
\end{proposition}

\begin{proof}
	Associativity follows from associativity of path concatenation, and the empty
	graphs are its identity elements. Additivity of the degree is immediate from
	the definition. Compatibility of the quantum weights follows from
	associativity of channel composition.
\end{proof}

An execution graph is terminating if its target belongs to
$L_{\mathrm{term}}$. In that case, its classical result is
\[
\operatorname{res}(\Gamma)=r(t(\Gamma)).
\]
Nonterminating behaviour is not represented by an infinite graph at this
stage. It appears through arbitrarily long finite prefixes and is recovered by
the directed-complete semantics.

\subsection{Finite unfoldings as prefix-closed graph families}
\label{subsec:finite-unfoldings}

\begin{definition}
	\label{def:finite-unfolding}
	Let $\ell_0\in L$ be an initial location. A finite unfolding rooted at
	$\ell_0$ is a finite set $E\subseteq\Exec(\Sigma)$ satisfying the following
	conditions.
	\begin{enumerate}
		\item Every $\Gamma\in E$ has source $\ell_0$.
		
		\item The empty graph $\mathbf 1_{\ell_0}$ belongs to $E$.
		
		\item The set $E$ is prefix-closed: if
		$\Gamma_1\star\Gamma_2\in E$, then $\Gamma_1\in E$.
		
		\item For every nonterminal graph $\Gamma\in E$ that is not maximal in $E$,
		there exists a command $c_{\Gamma}$ with
		\[
		s(c_{\Gamma})=t(\Gamma)
		\]
		such that the immediate extensions of $\Gamma$ in $E$ are precisely
		\[
		\Gamma\star
		\bigl(
		t(\Gamma)
		\xrightarrow{(c_{\Gamma},o)}
		\tau_{c_{\Gamma}}(o)
		\bigr),
		\qquad o\in O_{c_{\Gamma}}.
		\]
	\end{enumerate}
\end{definition}

Thus, a finite unfolding is equivalently a finite rooted execution tree, but
we identify it with the prefix-closed family of its path graphs. We denote by
$\operatorname{Term}(E)$ the set of terminating maximal graphs of $E$.
Maximal graphs ending at nonterminal locations are unresolved prefixes created
by the finite truncation and do not contribute to the terminating instrument.

\begin{definition}
	\label{def:unfolding-instrument}
	The terminating instrument associated with a finite unfolding $E$ is
	\[
	\mathcal I_E
	=
	\sum_{\Gamma\in\operatorname{Term}(E)}
	\delta\bigl(\operatorname{res}(\Gamma),\Phi_{\Gamma}\bigr).
	\]
\end{definition}

\begin{proposition}
	\label{prop:unfolding-instrument}
	For every finite unfolding $E$, the family of channels occurring in
	$\mathcal I_E$ has a subunital sum. Consequently, $\mathcal I_E$ is a finite
	quantum instrument.
\end{proposition}

\begin{proof}
	We argue by induction on the number of expanded nonterminal prefixes. If no
	nonterminal prefix is expanded, then either the root is terminal, in which
	case the only contribution is the identity channel, or the root is
	nonterminal and unresolved, in which case the terminating instrument is zero.
	Both cases are subunital.
	
	Assume that the root is expanded by a command $c$. For every $o\in O_c$, let
	$E_o$ denote the unfolding following outcome $o$, and let $\Psi_o$ be the sum
	of the channels attached to its terminating paths. By the induction
	hypothesis,
	\[
	\Psi_o(1_{\mathcal A})\leq 1_{\mathcal A}.
	\]
	Since $\Phi_{c,o}$ is positive,
	\[
	\Phi_{c,o}\bigl(\Psi_o(1_{\mathcal A})\bigr)
	\leq
	\Phi_{c,o}(1_{\mathcal A}).
	\]
	Summing over the outcomes gives
	\[
	\sum_{o\in O_c}
	\Phi_{c,o}\bigl(\Psi_o(1_{\mathcal A})\bigr)
	\leq
	\sum_{o\in O_c}\Phi_{c,o}(1_{\mathcal A})
	\leq
	1_{\mathcal A}.
	\]
	This is the required subunitality condition.
\end{proof}

Suppose that $E\subseteq E'$ and that $E'$ is obtained from $E$ by expanding
some unresolved maximal prefixes. Then
\[
\mathcal I_E\leq\mathcal I_{E'}
\]
in the completely positive order. Hence an increasing sequence of finite
unfoldings determines a directed family of finite instruments. Its relation
to the Kleene approximants of the corresponding recursive term will be
established in Sections~\ref{sec:graph-semantics} and~\ref{sec:abel}.

\subsection{The graded algebra of execution-graph series}
\label{subsec:graded-graph-series}

Let
\[
\C\langle\!\langle\Exec(\Sigma)\rangle\!\rangle
\]
denote the vector space of all formal sums
\[
a
=
\sum_{\Gamma\in\Exec(\Sigma)}a_{\Gamma}[\Gamma],
\qquad
a_{\Gamma}\in\C.
\]
The symbol $[\Gamma]$ records the execution graph and is not identified with
its quantum weight. Multiplication is first defined on basis elements by
\[
[\Gamma_1][\Gamma_2]
=
\begin{cases}
	[\Gamma_1\star\Gamma_2],
	& t(\Gamma_1)=s(\Gamma_2),\\
	0,
	& t(\Gamma_1)\neq s(\Gamma_2),
\end{cases}
\]
and is then extended by the Cauchy rule
\[
(ab)_{\Gamma}
=
\sum_{\Gamma=\Gamma_1\star\Gamma_2}
a_{\Gamma_1}b_{\Gamma_2}.
\]

\begin{proposition}
	\label{prop:graph-series-algebra}
	The Cauchy product above is well defined and associative. The space
	$\C\langle\!\langle\Exec(\Sigma)\rangle\!\rangle$ is a complete graded
	algebra, with homogeneous component of degree $n$ given by
	\[
	\mathscr A_n
	=
	\prod_{\substack{\Gamma\in\Exec(\Sigma)\\|\Gamma|=n}}
	\C[\Gamma].
	\]
	Its degree-zero component is generated by the empty graphs
	$[\mathbf 1_{\ell}]$, with $\ell\in L$.
\end{proposition}

\begin{proof}
	A graph of degree $n$ has exactly $n+1$ decompositions into an initial segment
	and a final segment. Therefore, every coefficient in the Cauchy product is a
	finite sum. Associativity follows from associativity of concatenation.
	Additivity of the degree under concatenation gives the graded decomposition.
	Completeness refers to the product decomposition
	\[
	\C\langle\!\langle\Exec(\Sigma)\rangle\!\rangle
	=
	\prod_{n\in\N}\mathscr A_n.
	\]
\end{proof}

This construction is a special case of formal series indexed by an
$\N$-graded small category, as considered in
\cite{Magnot2025FormalSeries}. No differential structure on the series space
is required in the present paper. We shall use only the grading, the local
finiteness of the Cauchy product, and completeness with respect to increasing
degree.

For a formal parameter $q$, define the degree-weighting operator
\[
R_q\left(
\sum_{\Gamma}a_{\Gamma}[\Gamma]
\right)
=
\sum_{\Gamma}q^{|\Gamma|}a_{\Gamma}[\Gamma].
\]
Since the degree is additive, $R_q$ is an algebra endomorphism:
\[
R_q(ab)=R_q(a)R_q(b).
\]
With the parametrisation
\[
q=e^{-t},
\qquad t>0,
\]
we write $R_t=R_{e^{-t}}$. Then
\[
R_{t+s}=R_tR_s,
\qquad s,t>0,
\]
so that $(R_t)_{t>0}$ is the semigroup generated by the graph-degree operator
\[
N[\Gamma]=|\Gamma|[\Gamma].
\]
The limit $t\to 0^{+}$, equivalently $q\to 1^{-}$, removes the exponential
suppression of long execution graphs. Its semantic meaning will be analysed
in Section~\ref{sec:abel}.

\subsection{The repeat-until-success graph family}
\label{subsec:rus-graphs}

Consider the repeat-until-success computation of
Section~\ref{subsec:rus}. It has one nonterminal location $\ell$ and one
terminal location $\ell_{\mathrm{done}}$. The unique command has outcomes
$\mathsf{false}$ and $\mathsf{true}$, with channels
\[
\Phi_{\mathsf{false}}
=
(1-p)\id_{\mathcal A},
\qquad
\Phi_{\mathsf{true}}
=
p\mathcal U.
\]
For every $n\in\N$, there is exactly one terminating execution graph with
$n$ failures followed by one success. Denote it by $\Gamma_n$. Its degree and
quantum weight are
\[
|\Gamma_n|=n+1
\]
and
\[
\Phi_{\Gamma_n}
=
p(1-p)^n\mathcal U.
\]
The associated formal execution series is
\[
\mathscr G_{\mathrm{RUS}}(q)
=
\sum_{n\geq 0}q^{n+1}[\Gamma_n].
\]
After evaluation of the graph weights, its channel-valued generating series
is
\[
\sum_{n\geq 0}q^{n+1}\Phi_{\Gamma_n}
=
\frac{pq}{1-(1-p)q}\mathcal U,
\qquad
|q|<(1-p)^{-1}.
\]
In particular, the Abel limit exists and satisfies
\[
\lim_{q\to 1^{-}}
\sum_{n\geq 0}q^{n+1}\Phi_{\Gamma_n}
=
\mathcal U,
\]
in agreement with the least-fixed-point calculation of
Section~\ref{subsec:rus}. The general relation between graph-series evaluation
and recursive denotational semantics is the subject of the next two sections.

\section{Compositional graph-series semantics}
\label{sec:graph-semantics}

The graph algebra introduced in Section~\ref{sec:graphs} records finite
execution histories independently of their denotational interpretation. We now
define an evaluation of execution graphs in the quantum orchestra monad and
show that graph concatenation and continuation grafting reproduce ordinary
channel composition and Kleisli composition, respectively. The resulting
construction gives a compositional semantics for finite unfoldings and, after
completion with respect to the graph degree, for locally finite execution
series.

Throughout this section, the quantum store is represented by a fixed von
Neumann algebra $\mathcal A$. All execution graphs are assumed to belong to a
finitary hybrid control system $\Sigma$ as in
Definition~\ref{def:control-system}.

\subsection{Evaluation of terminating graph polynomials}
\label{subsec:graph-polynomial-evaluation}

Let $X$ be the classical result set of $\Sigma$. A terminating graph
polynomial is a finite sum
\[
a
=
\sum_{\Gamma\in F}a_{\Gamma}[\Gamma],
\]
where $F$ is a finite set of terminating execution graphs and
$a_{\Gamma}\in\mathbb R_{+}$. We call $a$ semantically admissible if
\[
\sum_{\Gamma\in F}a_{\Gamma}\Phi_{\Gamma}
\]
is subunital.

\begin{definition}
	\label{def:graph-evaluation}
	The semantic evaluation of an admissible terminating graph polynomial is the
	finite quantum instrument
	\[
	\operatorname{ev}_{X}(a)
	=
	\sum_{\Gamma\in F}
	a_{\Gamma}\,
	\delta\bigl(\operatorname{res}(\Gamma),\Phi_{\Gamma}\bigr).
	\]
	Through the canonical embedding of finite instruments into quantum
	orchestras, we also regard
	\[
	\operatorname{ev}_{X}(a)
	\]
	as an element of $\QO(\mathcal A,\mathcal A,X)$.
\end{definition}

The admissibility condition is exactly the condition required for the sum in
Definition~\ref{def:graph-evaluation} to be a quantum instrument. In
particular, the characteristic polynomial of the terminating paths of a
finite unfolding,
\[
\mathscr G_{E}
=
\sum_{\Gamma\in\operatorname{Term}(E)}[\Gamma],
\]
is admissible by Proposition~\ref{prop:unfolding-instrument}, and
\[
\operatorname{ev}_{X}(\mathscr G_{E})
=
\mathcal I_{E}.
\]

The evaluation is compatible with nonnegative linear combinations whenever
the resulting polynomial remains admissible. More precisely, if $a$ and $b$
are admissible and if $\alpha,\beta\in\mathbb R_{+}$ satisfy
$\alpha+\beta\leq 1$, then
\[
\operatorname{ev}_{X}(\alpha a+\beta b)
=
\alpha\operatorname{ev}_{X}(a)
+
\beta\operatorname{ev}_{X}(b).
\]
This is the graph-level counterpart of the subconvexity axiom for quantum
orchestras.

\subsection{Open graph polynomials and channel composition}
\label{subsec:open-graph-polynomials}

For control locations $\ell,m\in L$, let
\[
\mathscr P_{+}(\ell,m)
\]
denote the cone of finite graph polynomials with nonnegative coefficients,
supported on execution graphs with source $\ell$ and target $m$. Their
evaluation is defined by
\[
\operatorname{ev}_{\ell,m}
\left(
\sum_{\Gamma\in F}a_{\Gamma}[\Gamma]
\right)
=
\sum_{\Gamma\in F}a_{\Gamma}\Phi_{\Gamma},
\]
whenever the right-hand side is a normal completely positive subunital map.
Thus,
\[
\operatorname{ev}_{\ell,m}(a)
\in
\Wstar(\mathcal A,\mathcal A).
\]

If $a\in\mathscr P_{+}(\ell,m)$ and
$b\in\mathscr P_{+}(m,n)$, their Cauchy product is supported on paths from
$\ell$ to $n$. The order of the factors agrees with execution order: the
paths represented by $a$ are executed first, and the paths represented by
$b$ are executed afterwards.

\begin{proposition}
	\label{prop:evaluation-composition}
	Let $a\in\mathscr P_{+}(\ell,m)$ and
	$b\in\mathscr P_{+}(m,n)$ be semantically admissible. Then $ab$ is
	semantically admissible and
	\[
	\operatorname{ev}_{\ell,n}(ab)
	=
	\operatorname{ev}_{\ell,m}(a)
	\circ
	\operatorname{ev}_{m,n}(b).
	\]
\end{proposition}

\begin{proof}
	Write
	\[
	a=\sum_{\Gamma}a_{\Gamma}[\Gamma]
	\qquad\text{and}\qquad
	b=\sum_{\Delta}b_{\Delta}[\Delta].
	\]
	By definition of the Cauchy product,
	\[
	ab
	=
	\sum_{\Gamma,\Delta}
	a_{\Gamma}b_{\Delta}[\Gamma\star\Delta].
	\]
	Using
	\[
	\Phi_{\Gamma\star\Delta}
	=
	\Phi_{\Gamma}\circ\Phi_{\Delta},
	\]
	we obtain
	\[
	\begin{aligned}
		\operatorname{ev}_{\ell,n}(ab)
		&=
		\sum_{\Gamma,\Delta}
		a_{\Gamma}b_{\Delta}
		\bigl(\Phi_{\Gamma}\circ\Phi_{\Delta}\bigr)
		\\
		&=
		\left(\sum_{\Gamma}a_{\Gamma}\Phi_{\Gamma}\right)
		\circ
		\left(\sum_{\Delta}b_{\Delta}\Phi_{\Delta}\right)
		\\
		&=
		\operatorname{ev}_{\ell,m}(a)
		\circ
		\operatorname{ev}_{m,n}(b).
	\end{aligned}
	\]
	The composition of two normal completely positive subunital maps is again
	normal, completely positive, and subunital. Hence $ab$ is semantically
	admissible.
\end{proof}

Thus, graph-series multiplication is not merely a combinatorial operation: it
is a refinement of sequential composition in the denotational model.

\subsection{Continuation grafting and Kleisli composition}
\label{subsec:continuation-grafting}

Sequential composition of hybrid computations is more general than ordinary
channel composition because the second computation may depend on the
classical result of the first one. We now express this dependence by grafting
continuation graphs onto terminating execution graphs.

Let $X$ and $Y$ be classical result sets. For every $x\in X$, choose an
initial control location $\kappa(x)$ for a continuation computation returning
values in $Y$. Let $a$ be an admissible terminating graph polynomial with
results in $X$, and let
\[
b_x
=
\sum_{\Delta}b_{x,\Delta}[\Delta]
\]
be an admissible terminating graph polynomial rooted at $\kappa(x)$ and with
results in $Y$.

For a terminating graph $\Gamma$ with
\[
\operatorname{res}(\Gamma)=x,
\]
and a graph $\Delta$ occurring in $b_x$, we denote by
\[
\Gamma\diamond\Delta
\]
the execution graph obtained by grafting $\Delta$ after $\Gamma$, identifying
the result interface $x$ with the continuation input $\kappa(x)$. Its degree,
quantum weight, and result are defined by
\[
|\Gamma\diamond\Delta|
=
|\Gamma|+|\Delta|,
\]
\[
\Phi_{\Gamma\diamond\Delta}
=
\Phi_{\Gamma}\circ\Phi_{\Delta},
\]
and
\[
\operatorname{res}(\Gamma\diamond\Delta)
=
\operatorname{res}(\Delta).
\]

\begin{definition}
	\label{def:continuation-grafting}
	The continuation grafting of $a$ with the family
	$b=(b_x)_{x\in X}$ is the graph polynomial
	\[
	a\diamond b
	=
	\sum_{\Gamma}
	\sum_{\Delta}
	a_{\Gamma}b_{\operatorname{res}(\Gamma),\Delta}
	[\Gamma\diamond\Delta].
	\]
\end{definition}

Assume that $X$ is a dcpo and define
\[
f:X\longrightarrow\QO(\mathcal A,\mathcal A,Y)
\]
by
\[
f(x)=\operatorname{ev}_{Y}(b_x).
\]

\begin{theorem}
	\label{thm:grafting-kleisli}
	Assume that $a$ and every $b_x$ are semantically admissible and that
	$f$ is Scott-continuous. Then $a\diamond b$ is semantically admissible and
	\[
	\operatorname{ev}_{Y}(a\diamond b)
	=
	\kleisli{f}\bigl(\operatorname{ev}_{X}(a)\bigr).
	\]
\end{theorem}

\begin{proof}
	Let
	\[
	a
	=
	\sum_{\Gamma}a_{\Gamma}[\Gamma].
	\]
	By Definition~\ref{def:graph-evaluation},
	\[
	\operatorname{ev}_{X}(a)
	=
	\sum_{\Gamma}
	a_{\Gamma}
	\delta\bigl(\operatorname{res}(\Gamma),\Phi_{\Gamma}\bigr).
	\]
	Kleisli extension is subconvex-linear on finite instruments. Hence
	\[
	\kleisli{f}\bigl(\operatorname{ev}_{X}(a)\bigr)
	=
	\sum_{\Gamma}
	a_{\Gamma}
	\kleisli{f}
	\left(
	\delta\bigl(\operatorname{res}(\Gamma),\Phi_{\Gamma}\bigr)
	\right).
	\]
	For a Dirac instrument, the Kleisli composition formula gives
	\[
	\kleisli{f}
	\left(
	\delta\bigl(\operatorname{res}(\Gamma),\Phi_{\Gamma}\bigr)
	\right)
	=
	\sum_{\Delta}
	b_{\operatorname{res}(\Gamma),\Delta}
	\delta
	\left(
	\operatorname{res}(\Delta),
	\Phi_{\Gamma}\circ\Phi_{\Delta}
	\right).
	\]
	Using the defining properties of the grafted graph
	$\Gamma\diamond\Delta$, the right-hand side becomes
	\[
	\sum_{\Delta}
	b_{\operatorname{res}(\Gamma),\Delta}
	\delta
	\left(
	\operatorname{res}(\Gamma\diamond\Delta),
	\Phi_{\Gamma\diamond\Delta}
	\right).
	\]
	Summing over $\Gamma$ yields exactly
	\[
	\operatorname{ev}_{Y}(a\diamond b).
	\]
	Since the right-hand side is a quantum orchestra, the graph polynomial
	$a\diamond b$ is semantically admissible.
\end{proof}

Theorem~\ref{thm:grafting-kleisli} is the fundamental compositionality result
for finite graph semantics. It states that grafting execution histories is a
combinatorial refinement of Kleisli composition in the quantum orchestra
monad.

\subsection{Degree weighting and compositionality}
\label{subsec:weighted-compositionality}

The graph-degree weighting introduced in
Section~\ref{subsec:graded-graph-series} is compatible with both concatenation
and continuation grafting. For a graph polynomial $a$, recall that
\[
R_q(a)
=
\sum_{\Gamma}q^{|\Gamma|}a_{\Gamma}[\Gamma].
\]
For a family $b=(b_x)_{x\in X}$, write
\[
R_q(b)
=
\bigl(R_q(b_x)\bigr)_{x\in X}.
\]

\begin{proposition}
	\label{prop:weighted-grafting}
	For every formal parameter $q$,
	\[
	R_q(a\diamond b)
	=
	R_q(a)\diamond R_q(b).
	\]
	Whenever the graph polynomials involved are semantically admissible, one has
	\[
	\operatorname{ev}_{Y}\bigl(R_q(a\diamond b)\bigr)
	=
	\kleisli{f_q}
	\left(
	\operatorname{ev}_{X}\bigl(R_q(a)\bigr)
	\right),
	\]
	where
	\[
	f_q(x)
	=
	\operatorname{ev}_{Y}\bigl(R_q(b_x)\bigr).
	\]
\end{proposition}

\begin{proof}
	For every pair of composable graphs,
	\[
	|\Gamma\diamond\Delta|
	=
	|\Gamma|+|\Delta|.
	\]
	Therefore,
	\[
	q^{|\Gamma\diamond\Delta|}
	=
	q^{|\Gamma|}q^{|\Delta|},
	\]
	which proves the first identity coefficient by coefficient. The second
	identity follows from Theorem~\ref{thm:grafting-kleisli} applied to the
	weighted graph polynomials.
\end{proof}

With $q=e^{-t}$, the proposition says that exponential damping by execution
length is compositional. The regularisation can therefore be applied locally
to program components before they are combined.

\subsection{Locally finite execution series}
\label{subsec:locally-finite-execution-series}

We now pass from finite graph polynomials to infinite execution series. A
terminating execution series with results in $X$ is a formal sum
\[
\mathscr G
=
\sum_{\Gamma}a_{\Gamma}[\Gamma],
\qquad
a_{\Gamma}\in\mathbb R_{+},
\]
supported on terminating execution graphs. It is called locally finite if,
for every $n\in\N$, only finitely many graphs of degree at most $n$ have a
nonzero coefficient. Its degree-$n$ truncation is
\[
\mathscr G_{\leq n}
=
\sum_{|\Gamma|\leq n}a_{\Gamma}[\Gamma].
\]

\begin{definition}
	\label{def:semantically-admissible-series}
	A locally finite terminating execution series $\mathscr G$ is semantically
	admissible if every truncation $\mathscr G_{\leq n}$ is admissible.
\end{definition}

For a semantically admissible series, the sequence
\[
\left(
\operatorname{ev}_{X}(\mathscr G_{\leq n})
\right)_{n\in\N}
\]
is increasing in the quantum orchestra order. We define
\[
\operatorname{ev}_{X}(\mathscr G)
=
\sup_{n\in\N}
\operatorname{ev}_{X}(\mathscr G_{\leq n}).
\]
The supremum exists because
$\QO(\mathcal A,\mathcal A,X)$ is a pointed dcpo.

For $q\in(0,1]$, define
\[
R_q(\mathscr G)
=
\sum_{\Gamma}q^{|\Gamma|}a_{\Gamma}[\Gamma].
\]
Since multiplication by $q^{|\Gamma|}$ decreases every nonnegative
coefficient, the truncations of $R_q(\mathscr G)$ remain admissible whenever
the truncations of $\mathscr G$ are admissible. We therefore obtain the
$q$-weighted denotation
\[
\mathcal Z_{\mathscr G}(q)
=
\operatorname{ev}_{X}\bigl(R_q(\mathscr G)\bigr).
\]
For $q=e^{-t}$, we also write
\[
\mathcal Z_{\mathscr G}(t)
=
\operatorname{ev}_{X}\bigl(R_t(\mathscr G)\bigr).
\]

\begin{proposition}
	\label{prop:monotonicity-in-q}
	Let $\mathscr G$ be a semantically admissible execution series. If
	\[
	0<q\leq q'\leq 1,
	\]
	then
	\[
	\mathcal Z_{\mathscr G}(q)
	\leq
	\mathcal Z_{\mathscr G}(q').
	\]
	Equivalently, the family
	\[
	t\longmapsto\mathcal Z_{\mathscr G}(t)
	\]
	is decreasing with respect to $t>0$.
\end{proposition}

\begin{proof}
	For every graph $\Gamma$,
	\[
	q^{|\Gamma|}
	\leq
	(q')^{|\Gamma|}.
	\]
	The difference between the corresponding finite truncations is therefore a
	finite sum of completely positive maps with nonnegative coefficients. Hence
	\[
	\operatorname{ev}_{X}
	\bigl(R_q(\mathscr G_{\leq n})\bigr)
	\leq
	\operatorname{ev}_{X}
	\bigl(R_{q'}(\mathscr G_{\leq n})\bigr).
	\]
	Taking directed suprema over $n$ proves the result.
\end{proof}

\subsection{Graph-series semantics of a program}
\label{subsec:program-graph-series}

Let $P$ be a closed finitary hybrid program with initial control location
$\ell_P$ and classical result set $X$. Denote by
\[
\operatorname{Term}(P)
\]
the set of all finite terminating execution graphs generated by $P$. The
formal execution series of $P$ is
\[
\mathscr G_P
=
\sum_{\Gamma\in\operatorname{Term}(P)}[\Gamma].
\]
Its degree-$n$ truncation is
\[
\mathscr G_{P,\leq n}
=
\sum_{\substack{\Gamma\in\operatorname{Term}(P)\\|\Gamma|\leq n}}
[\Gamma].
\]
For a finitary program, each truncation is finite. Moreover, it is the
characteristic polynomial of the terminating paths in the depth-$n$
unfolding of $P$. Consequently,
\[
\operatorname{ev}_{X}(\mathscr G_{P,\leq n})
=
\mathcal I_{E_n(P)},
\]
where $E_n(P)$ denotes the depth-$n$ unfolding.

\begin{definition}
	\label{def:program-graph-semantics}
	The graph-series denotation of $P$, whenever $\mathscr G_P$ is semantically
	admissible, is
	\[
	\sem{P}_{\mathrm{gr}}
	=
	\operatorname{ev}_{X}(\mathscr G_P).
	\]
	Its Abel-regularised graph denotation is
	\[
	\sem{P}_{q}
	=
	\operatorname{ev}_{X}\bigl(R_q(\mathscr G_P)\bigr),
	\qquad
	0<q<1.
	\]
	Equivalently, for $q=e^{-t}$,
	\[
	\sem{P}_{t}
	=
	\operatorname{ev}_{X}\bigl(R_t(\mathscr G_P)\bigr),
	\qquad
	t>0.
	\]
\end{definition}

The constructions above separate two issues that are conflated by the final
least-fixed-point denotation. The formal series $\mathscr G_P$ records the
finite execution histories and their degrees, while semantic evaluation
forgets the individual graphs and sums their quantum effects. The next
section proves that, for recursive programs generated by Scott-continuous
recursion functionals, the truncations of $\mathscr G_P$ coincide with the
Kleene approximants and that the limit $q\to 1^{-}$ reconstructs the ordinary
denotation.

\section{Kleene approximants and Abel reconstruction}
\label{sec:abel}

The graph-series semantics of Section~\ref{sec:graph-semantics} separates a
recursive computation into its finite terminating execution histories. We now
show that this decomposition is compatible with the least-fixed-point
semantics of quantum orchestras. The main result identifies the depth
truncations of the execution series with the Kleene approximants and proves
that exponential degree damping, with $q=e^{-t}$, reconstructs the recursive
denotation in the limit $t\to 0^{+}$.

The arguments are order-theoretic. No norm convergence is required for the
basic reconstruction theorem. Analytic convergence estimates are given at the
end of the section under additional Banach-space assumptions.

\subsection{The one-step semantic functional}
\label{subsec:one-step-functional}

Let $P$ be a closed finitary hybrid program represented by a control system
$\Sigma$ as in Definition~\ref{def:control-system}. We assume that the syntax
of $P$ selects a unique command $c_{\ell}$ at every reachable nonterminal
control location $\ell$. This does not restrict the branching caused by
quantum measurements: the command $c_{\ell}$ may still have several classical
outcomes. A program with several syntactic commands available at one location
can always be refined by replacing that location with sufficiently many
control locations.

Set
\[
L^{\circ}=L\setminus L_{\mathrm{term}}
\]
and consider the product dcpo
\[
\mathcal D_{P}
=
\prod_{\ell\in L^{\circ}}
\QO(\mathcal A,\mathcal A,X),
\]
equipped with the pointwise order. Its least element is denoted by
$\bottom_{P}$.

For
\[
\Phi\in\Wstar(\mathcal A,\mathcal A)
\]
and
\[
\xi\in\QO(\mathcal A,\mathcal A,X),
\]
define the left action of $\Phi$ on $\xi$ by
\[
(\Phi\triangleright\xi)_{k}
=
\Phi\circ\xi_{k}.
\]

\begin{lemma}
	\label{lem:channel-action}
	The assignment
	\[
	(\Phi,\xi)
	\longmapsto
	\Phi\triangleright\xi
	\]
	is well defined. For fixed $\Phi$, the map
	\[
	\xi\longmapsto\Phi\triangleright\xi
	\]
	is Scott-continuous and preserves the least element and directed suprema.
\end{lemma}

\begin{proof}
	Normality, complete positivity, and subunitality are preserved by composition.
	The continuity, subconvexity, and compositionality axioms of a quantum
	orchestra are also preserved by left composition with $\Phi$. If
	$(\xi_i)_{i\in I}$ is directed, then directed suprema in the hom-sets of
	$\Wstar$ are computed pointwise and composition is Scott-continuous. Hence
	\[
	\Phi\triangleright\left(\sup_{i\in I}\xi_i\right)
	=
	\sup_{i\in I}(\Phi\triangleright\xi_i).
	\]
	The statement concerning the least element is immediate.
\end{proof}

For $\boldsymbol\xi=(\xi_{\ell})_{\ell\in L^{\circ}}\in\mathcal D_P$, define
$F_P(\boldsymbol\xi)\in\mathcal D_P$ by
\[
F_P(\boldsymbol\xi)_{\ell}
=
\sum_{o\in O_{c_{\ell}}}
B_{\ell,o}(\boldsymbol\xi),
\]
where
\[
B_{\ell,o}(\boldsymbol\xi)
=
\begin{cases}
	\delta\bigl(r(\tau_{c_{\ell}}(o)),\Phi_{c_{\ell},o}\bigr),
	& \tau_{c_{\ell}}(o)\in L_{\mathrm{term}},
	\\[4pt]
	\Phi_{c_{\ell},o}\triangleright
	\xi_{\tau_{c_{\ell}}(o)},
	& \tau_{c_{\ell}}(o)\in L^{\circ}.
\end{cases}
\]
The sum is well defined because each continuation orchestra is subunital and
\[
\sum_{o\in O_{c_{\ell}}}
\Phi_{c_{\ell},o}(1_{\mathcal A})
\leq
1_{\mathcal A}.
\]

\begin{proposition}
	\label{prop:one-step-continuity}
	The map
	\[
	F_P:\mathcal D_P\longrightarrow\mathcal D_P
	\]
	is Scott-continuous. Consequently, it has a least fixed point
	\[
	\fix(F_P)
	=
	\sup_{n\in\N}F_P^{n}(\bottom_P).
	\]
\end{proposition}

\begin{proof}
	Every component of $F_P$ is a finite sum of constant Dirac orchestras and maps
	obtained by composing a coordinate projection with the Scott-continuous action
	of Lemma~\ref{lem:channel-action}. Finite sums preserve directed suprema in the
	present subconvex domain. Hence every component of $F_P$ is Scott-continuous,
	and so is $F_P$. The fixed-point formula follows from the Kleene fixed-point
	theorem.
\end{proof}

For an initial nonterminal location $\ell_P$, the ordinary recursive
denotation of the program is
\[
\sem{P}
=
\fix(F_P)_{\ell_P}.
\]
If the initial location is terminal, the denotation is simply
$\eta_X^{\mathcal A}(r(\ell_P))$ and no fixed-point construction is needed.

\subsection{Execution depth and Kleene approximation}
\label{subsec:depth-kleene}

For $\ell\in L^{\circ}$ and $n\in\N$, let
\[
\mathscr G_{P,\ell,\leq n}
=
\sum_{\substack{\Gamma\in\operatorname{Term}(P)\\
		s(\Gamma)=\ell,\ |\Gamma|\leq n}}
[\Gamma]
\]
be the polynomial formed by all terminating execution graphs starting at
$\ell$ and having degree at most $n$. We set
\[
A_{\ell,n}
=
\operatorname{ev}_{X}
\bigl(\mathscr G_{P,\ell,\leq n}\bigr).
\]
Since no terminating path can start at a nonterminal location and have degree
zero, one has
\[
A_{\ell,0}=\bottom.
\]

\begin{theorem}[Finite-unfolding correspondence]
	\label{thm:finite-unfolding-correspondence}
	For every $\ell\in L^{\circ}$ and every $n\in\N$,
	\[
	F_P^{n}(\bottom_P)_{\ell}
	=
	A_{\ell,n}.
	\]
	Thus, the $n$-th Kleene approximant is exactly the semantic evaluation of the
	terminating execution graphs of degree at most $n$.
\end{theorem}

\begin{proof}
	We argue by induction on $n$. For $n=0$, both sides are the least orchestra.
	Assume that the statement holds at depth $n$. A terminating path of degree at
	most $n+1$ starting at $\ell$ begins with an outcome $o$ of the command
	$c_{\ell}$. If $\tau_{c_{\ell}}(o)$ is terminal, this first transition already
	defines the Dirac orchestra
	\[
	\delta\bigl(r(\tau_{c_{\ell}}(o)),\Phi_{c_{\ell},o}\bigr).
	\]
	If $\tau_{c_{\ell}}(o)$ is nonterminal, the remainder of the path is a
	terminating path of degree at most $n$ starting at
	$\tau_{c_{\ell}}(o)$. By the induction hypothesis, the sum of the semantic
	weights of all such remainders is
	\[
	F_P^{n}(\bottom_P)_{\tau_{c_{\ell}}(o)}.
	\]
	Prepending the first transition acts by
	\[
	\Phi_{c_{\ell},o}\triangleright
	F_P^{n}(\bottom_P)_{\tau_{c_{\ell}}(o)}.
	\]
	Summing over the outcomes gives
	\[
	A_{\ell,n+1}
	=
	F_P\bigl(F_P^{n}(\bottom_P)\bigr)_{\ell}
	=
	F_P^{n+1}(\bottom_P)_{\ell}.
	\]
	This completes the induction.
\end{proof}

\begin{corollary}[Graph-series adequacy]
	\label{cor:graph-series-adequacy}
	For every initial nonterminal location $\ell_P$,
	\[
	\sem{P}_{\mathrm{gr}}
	=
	\sem{P}.
	\]
	More explicitly,
	\[
	\operatorname{ev}_{X}(\mathscr G_P)
	=
	\sup_{n\in\N}
	\operatorname{ev}_{X}(\mathscr G_{P,\ell_P,\leq n})
	=
	\fix(F_P)_{\ell_P}.
	\]
\end{corollary}

\begin{proof}
	By Definition~\ref{def:program-graph-semantics}, the graph-series denotation is
	the supremum of the denotations of its finite degree truncations. Theorem
	\ref{thm:finite-unfolding-correspondence} identifies these truncations with
	the Kleene approximants. Their supremum is the least fixed point of $F_P$.
\end{proof}

The result shows that the graph-series semantics is a conservative refinement
of the quantum orchestra semantics. It retains the degree and the individual
execution histories, but forgetting this extra information by semantic
evaluation recovers the original denotation.

\subsection{Exact-depth contributions}
\label{subsec:exact-depth-contributions}

For $n\in\N$, define the exact-depth graph polynomial
\[
\mathscr G_{P,\ell,n}
=
\sum_{\substack{\Gamma\in\operatorname{Term}(P)\\
		s(\Gamma)=\ell,\ |\Gamma|=n}}
[\Gamma]
\]
and its semantic contribution
\[
D_{\ell,n}
=
\operatorname{ev}_{X}(\mathscr G_{P,\ell,n}).
\]
The exact-depth contributions are positive in the quantum orchestra order and
satisfy
\[
A_{\ell,n}
=
\sum_{j=0}^{n}D_{\ell,j}.
\]
Equivalently,
\[
D_{\ell,n}
=
A_{\ell,n}-A_{\ell,n-1},
\qquad n\geq 1,
\]
where the difference is taken pointwise in the ambient vector spaces of
normal linear maps and is completely positive. We use the graph definition of
$D_{\ell,n}$ as the primary one, so no subtraction operation is required in
the dcpo itself.

For $0<q<1$, define the degree-damped denotation
\[
Z_{\ell}(q)
=
\operatorname{ev}_{X}
\bigl(R_q(\mathscr G_{P,\ell})\bigr)
=
\sup_{N\in\N}
\sum_{n=0}^{N}q^{n}D_{\ell,n}.
\]
Every finite partial sum is subunital because it is bounded above by the
corresponding unweighted truncation. Hence the directed supremum is a
well-defined quantum orchestra.

\begin{proposition}[Abel identity]
	\label{prop:abel-identity}
	For every $0<q<1$,
	\[
	Z_{\ell}(q)
	=
	(1-q)
	\sum_{n=0}^{\infty}q^{n}A_{\ell,n},
	\]
	where the series on the right is defined as the directed supremum of its
	finite partial sums.
\end{proposition}

\begin{proof}
	Using
	\[
	A_{\ell,n}
	=
	\sum_{j=0}^{n}D_{\ell,j},
	\]
	positivity allows us to rearrange the directed sums and obtain
	\[
	\begin{aligned}
		(1-q)
		\sum_{n=0}^{\infty}q^{n}A_{\ell,n}
		&=
		(1-q)
		\sum_{n=0}^{\infty}q^{n}
		\sum_{j=0}^{n}D_{\ell,j}
		\\
		&=
		\sum_{j=0}^{\infty}
		\left(
		(1-q)
		\sum_{n=j}^{\infty}q^{n}
		\right)
		D_{\ell,j}
		\\
		&=
		\sum_{j=0}^{\infty}q^{j}D_{\ell,j}
		\\
		&=
		Z_{\ell}(q).
	\end{aligned}
	\]
	All rearrangements are justified by monotonicity: each side is the supremum
	of the same directed family of finite positive sums.
\end{proof}

Proposition~\ref{prop:abel-identity} distinguishes two generating series. The
increment series
\[
\sum_{n=0}^{\infty}q^{n}D_{\ell,n}
\]
has a direct semantic interpretation as the degree-weighted sum of terminating
histories. The cumulative series
\[
\sum_{n=0}^{\infty}q^{n}A_{\ell,n}
\]
contains the same contribution repeatedly and therefore carries a universal
geometric divergence as $q\to 1^{-}$. The factor $1-q$ is the canonical Abel
normalisation that removes this repeated counting in the sense of classical
Abel summation \cite{Hardy1949}.

\subsection{Order-theoretic Abel reconstruction}
\label{subsec:order-abel-reconstruction}

We now prove that the regularised denotations recover the least-fixed-point
semantics without any norm-convergence assumption.

\begin{theorem}[Abel reconstruction]
	\label{thm:abel-reconstruction}
	Let $P$ be a closed finitary hybrid program and let $\ell\in L^{\circ}$. Then
	the family
	\[
	\bigl(Z_{\ell}(q)\bigr)_{0<q<1}
	\]
	is increasing in $q$ and bounded above by $\fix(F_P)_{\ell}$. Moreover,
	\[
	\sup_{0<q<1}Z_{\ell}(q)
	=
	\fix(F_P)_{\ell}.
	\]
	Equivalently,
	\[
	\lim_{q\to 1^{-}}^{\mathrm{Scott}}Z_{\ell}(q)
	=
	\fix(F_P)_{\ell}.
	\]
\end{theorem}

\begin{proof}
	If $0<q\leq q'<1$, then
	\[
	q^{n}D_{\ell,n}
	\leq
	(q')^{n}D_{\ell,n}
	\]
	for every $n$, and therefore
	\[
	Z_{\ell}(q)
	\leq
	Z_{\ell}(q').
	\]
	Since $q^{n}\leq 1$, one also has
	\[
	Z_{\ell}(q)
	\leq
	\sup_{N\in\N}
	\sum_{n=0}^{N}D_{\ell,n}
	=
	\sup_{N\in\N}A_{\ell,N}
	=
	\fix(F_P)_{\ell}.
	\]
	
	Let
	\[
	Z_{\ell}^{*}
	=
	\sup_{0<q<1}Z_{\ell}(q).
	\]
	Fix $N\in\N$. For every $0<q<1$, positivity gives
	\[
	Z_{\ell}(q)
	\geq
	\sum_{n=0}^{N}q^{n}D_{\ell,n}
	\geq
	q^{N}
	\sum_{n=0}^{N}D_{\ell,n}
	=
	q^{N}A_{\ell,N}.
	\]
	Taking the supremum over $q\in(0,1)$ yields
	\[
	Z_{\ell}^{*}
	\geq
	A_{\ell,N},
	\]
	because
	\[
	\sup_{0<q<1}q^{N}A_{\ell,N}
	=
	A_{\ell,N}.
	\]
	Since this holds for every $N$,
	\[
	Z_{\ell}^{*}
	\geq
	\sup_{N\in\N}A_{\ell,N}
	=
	\fix(F_P)_{\ell}.
	\]
	Combined with the opposite inequality proved above, this gives the result.
\end{proof}

\begin{corollary}
	\label{cor:pointwise-ultraweak-abel}
	Let $\mathcal K$ be a von Neumann algebra, let
	\[
	k:X\longrightarrow\Wstar(\mathcal K,\mathcal A)
	\]
	be a Scott-continuous continuation, and let $b\in\mathcal K$ be positive.
	Then
	\[
	\bigl(Z_{\ell}(q)_{k}(b)\bigr)_{0<q<1}
	\]
	is an increasing net in $\mathcal A_{+}$ and
	\[
	Z_{\ell}(q)_{k}(b)
	\longrightarrow
	\bigl(\fix(F_P)_{\ell}\bigr)_{k}(b)
	\]
	ultraweakly as $q\to 1^{-}$.
\end{corollary}

\begin{proof}
	The order on quantum orchestras is pointwise, and directed suprema in the
	hom-sets of $\Wstar$ are computed pointwise in the ultraweak topology. The
	claim therefore follows directly from Theorem~\ref{thm:abel-reconstruction}.
\end{proof}

\subsection{The heat parameter $q=e^{-t}$}
\label{subsec:heat-parameter}

Set
\[
q=e^{-t},
\qquad t>0,
\]
and write
\[
Z_{\ell}(t)
=
Z_{\ell}(e^{-t}).
\]
The exact-depth expansion becomes
\[
Z_{\ell}(t)
=
\sum_{n=0}^{\infty}e^{-tn}D_{\ell,n},
\]
while the Abel identity takes the form
\[
Z_{\ell}(t)
=
(1-e^{-t})
\sum_{n=0}^{\infty}e^{-tn}A_{\ell,n}.
\]
The factor $e^{-tn}$ suppresses long execution histories, and the limit
$t\to 0^{+}$ removes this suppression. Theorem~\ref{thm:abel-reconstruction}
can therefore be rewritten as
\[
\lim_{t\to 0^{+}}^{\mathrm{Scott}}Z_{\ell}(t)
=
\fix(F_P)_{\ell}.
\]

Since
\[
1-e^{-t}
\sim
t
\qquad
\text{as }t\to 0^{+},
\]
the cumulative generating series is renormalised by a factor asymptotic to
$t$. This renormalisation is not an additional choice: it is forced by the
relation between exact-depth contributions and cumulative Kleene
approximants.

\subsection{Norm estimates under geometric convergence}
\label{subsec:abel-norm-estimates}

The order-theoretic theorem is sufficient for denotational semantics. In some
applications, however, the relevant orchestras are represented in a Banach
space of linear maps and the Kleene approximants converge in norm. The Abel
reconstruction then also holds in norm, with an explicit estimate.

\begin{proposition}
	\label{prop:abel-norm-estimate}
	Suppose that the sequence $(A_{\ell,n})_{n\in\N}$ and its supremum
	\[
	A_{\ell}=\fix(F_P)_{\ell}
	\]
	belong to a Banach space with norm $\lVert\cdot\rVert$. Assume that there
	exist constants $C>0$ and $0\leq r<1$ such that
	\[
	\lVert A_{\ell}-A_{\ell,n}\rVert
	\leq
	Cr^{n}
	\]
	for every $n\in\N$. Then, for $0<q<1$,
	\[
	\lVert A_{\ell}-Z_{\ell}(q)\rVert
	\leq
	C\frac{1-q}{1-qr}.
	\]
	Consequently,
	\[
	Z_{\ell}(q)
	\longrightarrow
	A_{\ell}
	\]
	in norm as $q\to 1^{-}$.
\end{proposition}

\begin{proof}
	Since
	\[
	(1-q)\sum_{n=0}^{\infty}q^{n}=1,
	\]
	Proposition~\ref{prop:abel-identity} gives
	\[
	A_{\ell}-Z_{\ell}(q)
	=
	(1-q)
	\sum_{n=0}^{\infty}q^{n}
	\bigl(A_{\ell}-A_{\ell,n}\bigr).
	\]
	Therefore,
	\[
	\begin{aligned}
		\lVert A_{\ell}-Z_{\ell}(q)\rVert
		&\leq
		(1-q)
		\sum_{n=0}^{\infty}q^{n}
		\lVert A_{\ell}-A_{\ell,n}\rVert
		\\
		&\leq
		C(1-q)
		\sum_{n=0}^{\infty}(qr)^{n}
		\\
		&=
		C\frac{1-q}{1-qr}.
	\end{aligned}
	\]
	The right-hand side tends to zero as $q\to 1^{-}$.
\end{proof}

\subsection{Repeat-until-success revisited}
\label{subsec:rus-abel}

For the repeat-until-success computation of Sections~\ref{subsec:rus} and
\ref{subsec:rus-graphs}, the exact-depth contributions are
\[
D_{0}=0
\]
and
\[
D_{n}
=
p(1-p)^{n-1}\mathcal U,
\qquad n\geq 1.
\]
The Kleene approximants are
\[
A_n
=
\sum_{j=0}^{n}D_j
=
\bigl(1-(1-p)^n\bigr)\mathcal U.
\]
Consequently,
\[
\begin{aligned}
	Z(q)
	&=
	\sum_{n=1}^{\infty}q^{n}p(1-p)^{n-1}\mathcal U
	\\
	&=
	\frac{pq}{1-(1-p)q}\mathcal U
	\\
	&=
	(1-q)
	\sum_{n=0}^{\infty}
	q^{n}
	\bigl(1-(1-p)^n\bigr)\mathcal U.
\end{aligned}
\]
It follows that
\[
\lim_{q\to 1^{-}}Z(q)
=
\mathcal U.
\]
With $q=e^{-t}$, this becomes
\[
Z(t)
=
\frac{pe^{-t}}{1-(1-p)e^{-t}}\mathcal U,
\]
and
\[
\lim_{t\to 0^{+}}Z(t)
=
\mathcal U.
\]
Thus, in this elementary example, the graph expansion, the Abel
regularisation, and the least-fixed-point semantics can all be computed
explicitly and give the same result.

\section{Algebraic feedback and execution resolvents}
\label{sec:feedback}

The preceding sections reconstruct recursive denotations from finite execution
histories by means of degree-weighted graph series. We now isolate the
algebraic mechanism underlying a feedback loop. In a linear sector, one pass
through an open computation is represented by an operator
\[
T:C\longrightarrow R,
\]
while the return of the result to the continuation interface is represented by
an operator
\[
S:R\longrightarrow C.
\]
One complete traversal of the loop is therefore governed by the endomorphism
\[
K=ST:C\longrightarrow C.
\]
The regularised feedback equation involves the operator
\[
I_C-qST,
\]
where the factor $q$ records one additional loop traversal. Its inverse, when
it exists, resums all finite returns through the feedback interface. We first
construct this inverse formally, then relate it to least fixed points and to an
algebraic cross-ratio.

The present section is independent of any differential or Grassmannian
structure. Only linear algebra, formal series, and, in the analytic part,
standard operator theory are used. The relation with traced feedback and the
geometry of interaction is discussed in
\cite{JoyalStreetVerity1996,Hasegawa1997,HaghverdiScott2006,HasuoHoshino2017}.

\subsection{Linear feedback configurations}
\label{subsec:linear-feedback}

Let $C$ and $R$ be vector spaces over $\C$. We interpret $C$ as a continuation
interface and $R$ as a result interface.

\begin{definition}
	\label{def:feedback-configuration}
	A linear feedback configuration is a quadruple
	\[
	(C,R,T,S),
	\]
	where
	\[
	T:C\longrightarrow R
	\qquad\text{and}\qquad
	S:R\longrightarrow C
	\]
	are linear maps. Its return operator is
	\[
	K=ST\in\operatorname{End}(C).
	\]
	For a scalar parameter $q$, the associated feedback operator is
	\[
	D_q(T,S)=I_C-qST.
	\]
\end{definition}

Given an initial continuation value $u\in C$, the $q$-regularised feedback
equation is
\[
c=u+qSTc.
\]
Whenever $D_q(T,S)$ is invertible, this equation has the unique solution
\[
c=D_q(T,S)^{-1}u.
\]
The corresponding result at the output interface is
\[
r=TD_q(T,S)^{-1}u.
\]
This motivates the following notation.

\begin{definition}
	\label{def:feedback-response}
	Whenever $I_C-qST$ is invertible, the feedback response is the linear map
	\[
	\operatorname{Fb}_q(T,S)
	=
	T(I_C-qST)^{-1}:C\longrightarrow R.
	\]
\end{definition}

The parameter $q$ is central. In the simplest grading, one complete return
through the loop has degree one, and the term of degree $n$ represents an
execution that crosses the feedback interface exactly $n$ times. When the
degree counts elementary execution steps instead, the scalar operator $qST$
is replaced by a positively graded return series $K(q)$, as described below.

\subsection{Formal execution resolvents}
\label{subsec:formal-execution-resolvents}

We first work in the formal power-series algebra
\[
\operatorname{End}(C)[[q]].
\]
No topology on $C$ is needed.

\begin{proposition}[Formal Neumann resolvent]
	\label{prop:formal-neumann-resolvent}
	Let $K\in\operatorname{End}(C)$. Then $I_C-qK$ is invertible in
	$\operatorname{End}(C)[[q]]$, with inverse
	\[
	\mathcal R_K(q)
	=
	(I_C-qK)^{-1}
	=
	\sum_{n=0}^{\infty}q^nK^n.
	\]
	Consequently,
	\[
	\operatorname{Fb}_q(T,S)
	=
	\sum_{n=0}^{\infty}q^nT(ST)^n
	=
	\sum_{n=0}^{\infty}q^n(TS)^nT
	\]
	as an element of $\operatorname{Hom}(C,R)[[q]]$.
\end{proposition}

\begin{proof}
	For every $N\in\N$,
	\[
	(I_C-qK)
	\sum_{n=0}^{N}q^nK^n
	=
	I_C-q^{N+1}K^{N+1}.
	\]
	Hence the coefficient of every fixed power of $q$ agrees with the
	corresponding coefficient of the identity when $N$ is sufficiently large.
	This proves
	\[
	(I_C-qK)
	\sum_{n=0}^{\infty}q^nK^n
	=
	I_C.
	\]
	The same computation on the other side gives the two-sided inverse. Finally,
	\[
	T(ST)^n=(TS)^nT
	\]
	for every $n\in\N$, which proves the formula for the feedback response.
\end{proof}

The construction extends directly to a return operator that is itself a
graded execution series.

\begin{proposition}[Graded formal resolvent]
	\label{prop:graded-formal-resolvent}
	Let
	\[
	K(q)=\sum_{n=1}^{\infty}q^nK_n
	\in
	q\operatorname{End}(C)[[q]].
	\]
	Then $I_C-K(q)$ is invertible, with
	\[
	(I_C-K(q))^{-1}
	=
	\sum_{m=0}^{\infty}K(q)^m.
	\]
	For $N\geq 1$, the coefficient of $q^N$ in this inverse is
	\[
	\sum_{m=1}^{N}
	\sum_{\substack{n_1+\cdots+n_m=N\\ n_1,\ldots,n_m\geq 1}}
	K_{n_1}K_{n_2}\cdots K_{n_m}.
	\]
\end{proposition}

\begin{proof}
	Since $K(q)$ has strictly positive valuation, the coefficient of $q^N$ in
	$K(q)^m$ vanishes whenever $m>N$. Thus every coefficient of
	\[
	\sum_{m=0}^{\infty}K(q)^m
	\]
	is a finite sum. The usual geometric-series computation is therefore valid
	coefficient by coefficient. Expanding $K(q)^m$ gives the stated coefficient
	formula.
\end{proof}

Suppose that $K_n$ is the total semantic weight of all elementary return
cycles of degree $n$. Then the ordered product
\[
K_{n_1}K_{n_2}\cdots K_{n_m}
\]
is the weight of a feedback execution obtained by concatenating $m$ return
cycles of respective degrees $n_1,\ldots,n_m$. Therefore,
Proposition~\ref{prop:graded-formal-resolvent} states that the inverse of
$I_C-K(q)$ is precisely the formal sum of all finite feedback executions.
Local finiteness at each total degree is inherited from the positive valuation
of $K(q)$.

\subsection{Analytic resolvents}
\label{subsec:analytic-resolvents}

Assume now that $C$ and $R$ are Banach spaces and that $T$ and $S$ are bounded
operators. The formal resolvent becomes an analytic operator-valued function
inside its disk of convergence.

\begin{proposition}
	\label{prop:analytic-neumann-resolvent}
	Let $K\in\mathcal L(C)$, where $\mathcal L(C)$ denotes the Banach algebra of
	bounded operators on $C$. If
	\[
	|q|\,\lVert K\rVert<1,
	\]
	then
	\[
	(I_C-qK)^{-1}
	=
	\sum_{n=0}^{\infty}q^nK^n
	\]
	with convergence in operator norm. More generally, $I_C-qK$ is invertible if
	and only if
	\[
	q^{-1}\notin\Spec(K)
	\]
	for $q\neq 0$.
\end{proposition}

\begin{proof}
	The first statement is the Neumann-series theorem in the Banach algebra
	$\mathcal L(C)$. For $q\neq 0$, one has
	\[
	I_C-qK
	=
	-q\bigl(K-q^{-1}I_C\bigr),
	\]
	so invertibility is equivalent to $q^{-1}$ belonging to the resolvent set of
	$K$.
\end{proof}

With the parametrisation
\[
q=e^{-t},
\qquad t>0,
\]
the feedback operator and its resolvent become
\[
D_t(T,S)
=
I_C-e^{-t}ST
\]
and
\[
\mathcal R_t(T,S)
=
(I_C-e^{-t}ST)^{-1}.
\]
The damping factor $e^{-tn}$ assigns an exponential penalty to an execution
that traverses the feedback loop $n$ times.

\begin{proposition}
	\label{prop:resolvent-limit}
	Assume that $I_C-ST$ is invertible. Then there exists $t_0>0$ such that
	$I_C-e^{-t}ST$ is invertible for every $t\in[0,t_0)$, and
	\[
	\lim_{t\to 0^{+}}
	(I_C-e^{-t}ST)^{-1}
	=
	(I_C-ST)^{-1}
	\]
	in operator norm. If, in addition, the spectral radius satisfies
	\[
	r(ST)<1,
	\]
	then
	\[
	(I_C-ST)^{-1}
	=
	\sum_{n=0}^{\infty}(ST)^n
	\]
	in operator norm.
\end{proposition}

\begin{proof}
	The invertible group of a unital Banach algebra is open, and inversion is
	continuous. Since
	\[
	I_C-e^{-t}ST
	\longrightarrow
	I_C-ST
	\]
	in operator norm as $t\to 0^{+}$, the first assertion follows. If
	$r(ST)<1$, the spectral-radius formula implies that the Neumann series
	converges in operator norm.
\end{proof}

The distinction between the two assumptions in
Proposition~\ref{prop:resolvent-limit} is important. Invertibility of
$I_C-ST$ guarantees a regular limit of the analytic resolvent, but the
unweighted execution series
\[
\sum_{n=0}^{\infty}(ST)^n
\]
converges in operator norm only under the stronger spectral-radius condition.

\subsection{Feedback as a least fixed point}
\label{subsec:feedback-fixed-point}

The execution resolvent has a direct domain-theoretic interpretation. Let $C$
be a pointed dcpo cone whose least element is $0$. Assume that addition and
multiplication by scalars in $[0,1]$ preserve directed suprema in each
variable. Let
\[
K:C\longrightarrow C
\]
be Scott-continuous, additive, and positively homogeneous, and let $u\in C$.

For $0\leq q\leq 1$, define
\[
F_{q,u}(c)
=
u+qKc.
\]

\begin{proposition}
	\label{prop:feedback-kleene-chain}
	The map $F_{q,u}$ is Scott-continuous, and its Kleene approximants satisfy
	\[
	F_{q,u}^{n}(\bottom)
	=
	\sum_{j=0}^{n-1}q^jK^ju
	\]
	for every $n\geq 1$. Consequently,
	\[
	\fix(F_{q,u})
	=
	\sup_{n\geq 1}
	\sum_{j=0}^{n-1}q^jK^ju.
	\]
	If $C$ is embedded in a vector space in which $I_C-qK$ is invertible and the
	right-hand side converges to an element of $C$, then
	\[
	\fix(F_{q,u})
	=
	(I_C-qK)^{-1}u.
	\]
\end{proposition}

\begin{proof}
	Scott-continuity follows from the assumptions on addition, scalar
	multiplication, and $K$. The formula for the Kleene approximants is proved by
	induction. For $n=1$,
	\[
	F_{q,u}(\bottom)=u,
	\]
	because $K(0)=0$. If the formula holds for $n$, then
	\[
	\begin{aligned}
		F_{q,u}^{n+1}(\bottom)
		&=
		u+qK
		\left(
		\sum_{j=0}^{n-1}q^jK^ju
		\right)
		\\
		&=
		\sum_{j=0}^{n}q^jK^ju.
	\end{aligned}
	\]
	The fixed-point formula follows from the Kleene theorem. Finally, any limit
	$c$ of the partial sums satisfies
	\[
	c=u+qKc,
	\]
	and therefore
	\[
	(I_C-qK)c=u.
	\]
	Invertibility gives the last formula.
\end{proof}

In the setting of Sections~\ref{sec:graph-semantics} and~\ref{sec:abel}, the
terms $K^ju$ are the semantic contributions of executions that return through
the feedback interface exactly $j$ times. Proposition
\ref{prop:feedback-kleene-chain} therefore identifies the formal Neumann
series with the Kleene chain of the linear feedback functional.

\subsection{The cross-ratio of two graph decompositions}
\label{subsec:feedback-cross-ratio}

The feedback operator $I_C-qST$ also has an intrinsic algebraic description.
Set
\[
V=C\oplus R
\]
and identify $C$ and $R$ with the coordinate subspaces
\[
C_0=C\oplus\{0\},
\qquad
R_0=\{0\}\oplus R.
\]
For a linear map $T:C\to R$, define
\[
\Gamma(T)
=
\{(c,Tc)\mid c\in C\},
\]
and, for $S:R\to C$, define
\[
\Gamma(S)
=
\{(Sr,r)\mid r\in R\}.
\]
One always has
\[
V=\Gamma(T)\oplus R_0
\qquad\text{and}\qquad
V=C_0\oplus\Gamma(S).
\]

If $V=A\oplus B$, let
\[
p_{A,B}:V\longrightarrow A
\]
denote the projection onto $A$ parallel to $B$.

\begin{definition}
	\label{def:algebraic-cross-ratio}
	The algebraic cross-ratio associated with the two graph decompositions is the
	endomorphism of $C_0$ defined by
	\[
	\CR(C_0,R_0;\Gamma(T),\Gamma(S))
	=
	p_{C_0,\Gamma(S)}
	\circ
	p_{\Gamma(T),R_0}
	\big|_{C_0}.
	\]
	Through the canonical identification $C_0\simeq C$, it is regarded as an
	endomorphism of $C$.
\end{definition}

\begin{proposition}
	\label{prop:cross-ratio-feedback}
	For every pair of linear maps $T:C\to R$ and $S:R\to C$,
	\[
	\CR(C_0,R_0;\Gamma(T),\Gamma(S))
	=
	I_C-ST.
	\]
	More generally,
	\[
	\CR(C_0,R_0;\Gamma(qT),\Gamma(S))
	=
	I_C-qST.
	\]
\end{proposition}

\begin{proof}
	For $c\in C$, identified with $(c,0)\in C_0$, one has
	\[
	p_{\Gamma(T),R_0}(c,0)
	=
	(c,Tc).
	\]
	Moreover,
	\[
	(c,Tc)
	=
	(c-STc,0)+(STc,Tc),
	\]
	and the second term belongs to $\Gamma(S)$. Therefore,
	\[
	p_{C_0,\Gamma(S)}(c,Tc)
	=
	(c-STc,0).
	\]
	This proves the first formula. Replacing $T$ with $qT$ gives the second one.
\end{proof}

The definition above is purely algebraic. The more general smooth cross-ratio
construction studied in \cite{Magnot2024CrossRatio} is not needed for the
present semantic application.

\begin{theorem}[Transversality criterion]
	\label{thm:transversality-criterion}
	The following conditions are equivalent:
	\begin{enumerate}
		\item $V=\Gamma(qT)\oplus\Gamma(S)$;
		\item $I_C-qST$ is invertible on $C$;
		\item $I_R-qTS$ is invertible on $R$.
	\end{enumerate}
	When these conditions hold,
	\[
	(I_R-qTS)^{-1}
	=
	I_R+qT(I_C-qST)^{-1}S
	\]
	and
	\[
	(I_C-qST)^{-1}
	=
	I_C+qS(I_R-qTS)^{-1}T.
	\]
\end{theorem}

\begin{proof}
	Given $(c,r)\in C\oplus R$, a decomposition
	\[
	(c,r)
	=
	(x,qTx)+(Sy,y)
	\]
	is equivalent to
	\[
	c=x+Sy,
	\qquad
	r=qTx+y.
	\]
	Eliminating $y$ gives
	\[
	(I_C-qST)x
	=
	c-Sr.
	\]
	Thus every vector has a unique such decomposition if and only if
	$I_C-qST$ is invertible. This proves the equivalence of the first two
	conditions. The equivalence with the third condition and the inverse formulas
	follow from direct multiplication. For example,
	\[
	(I_R-qTS)
	\left(I_R+qT(I_C-qST)^{-1}S\right)
	=
	I_R.
	\]
	The reverse product is verified in the same way, and the second identity is
	symmetric.
\end{proof}

\begin{corollary}
	\label{cor:cross-ratio-resolvent}
	Whenever the graph subspaces $\Gamma(qT)$ and $\Gamma(S)$ are complementary,
	the inverse cross-ratio is the execution resolvent:
	\[
	\CR(C_0,R_0;\Gamma(qT),\Gamma(S))^{-1}
	=
	(I_C-qST)^{-1}.
	\]
	Consequently,
	\[
	\operatorname{Fb}_q(T,S)
	=
	T\,
	\CR(C_0,R_0;\Gamma(qT),\Gamma(S))^{-1}.
	\]
\end{corollary}

Thus, loss of invertibility of the feedback operator is exactly the failure of
transversality between the two graph subspaces. This algebraic observation
will be used in Section~\ref{sec:fredholm} to define a determinant that detects
singular feedback configurations.

\subsection{Covariance under changes of interface coordinates}
\label{subsec:feedback-covariance}

A feedback invariant should not depend on the choice of coordinates used to
represent the continuation and result interfaces. Let
\[
U:C\longrightarrow C'
\qquad\text{and}\qquad
V:R\longrightarrow R'
\]
be linear isomorphisms, and define
\[
T'=VTU^{-1},
\qquad
S'=USV^{-1}.
\]

\begin{proposition}
	\label{prop:feedback-covariance}
	For every scalar $q$,
	\[
	I_{C'}-qS'T'
	=
	U(I_C-qST)U^{-1}.
	\]
	Whenever these operators are invertible,
	\[
	(I_{C'}-qS'T')^{-1}
	=
	U(I_C-qST)^{-1}U^{-1}
	\]
	and
	\[
	\operatorname{Fb}_q(T',S')
	=
	V\operatorname{Fb}_q(T,S)U^{-1}.
	\]
\end{proposition}

\begin{proof}
	The first identity follows from
	\[
	S'T'
	=
	USV^{-1}VTU^{-1}
	=
	USTU^{-1}.
	\]
	The remaining identities follow by inversion and multiplication by $T'$.
\end{proof}

Hence the execution resolvent and the feedback response are covariant under
interface isomorphisms. In particular, spectral and determinant quantities
constructed from $I_C-qST$ depend only on the represented feedback semantics,
not on a choice of bases. The trace-class sector and the corresponding
Fredholm determinant are studied next.

\section{Fredholm invariants of feedback}
\label{sec:fredholm}

The execution resolvent of Section~\ref{sec:feedback} detects whether a
linear feedback equation has a unique solution. In a Hilbert-space sector, a
stronger scalar invariant is available. If the return operator is trace
class, the Fredholm determinant of the feedback operator records its
singularities, is invariant under changes of interface coordinates, and
admits an expansion in traces of closed loop iterates. This section develops
that construction.

No integrable-systems interpretation is assumed. In particular, the
Fredholm determinant introduced below is a feedback invariant and is not
claimed to satisfy Pl\"ucker or Hirota identities.

\subsection{The trace-class feedback sector}
\label{subsec:trace-class-feedback}

Let $C$ and $R$ be separable complex Hilbert spaces. We denote by
$\mathcal L(C,R)$ the space of bounded operators from $C$ to $R$, by
$\Stwo(C,R)$ the Hilbert--Schmidt class, and by $\Sone(C)$ the trace class
on $C$. We use the standard inclusions and ideal properties
\[
\Sone(C)
\subseteq
\Stwo(C)
\subseteq
\mathcal L(C)
\]
and
\[
\Stwo(R,C)\,\Stwo(C,R)
\subseteq
\Sone(C).
\]
We refer to \cite{GohbergKrein1969,Simon2005} for the theory of trace ideals
and Fredholm determinants.

\begin{assumption}
	\label{ass:hilbert-schmidt-feedback}
	The feedback configuration $(C,R,T,S)$ satisfies
	\[
	T\in\Stwo(C,R)
	\qquad\text{and}\qquad
	S\in\Stwo(R,C).
	\]
\end{assumption}

Under Assumption~\ref{ass:hilbert-schmidt-feedback}, the return operators
\[
K_C=ST\in\Sone(C)
\qquad\text{and}\qquad
K_R=TS\in\Sone(R)
\]
are trace class.

\begin{definition}
	\label{def:feedback-determinant}
	For $q\in\C$, the Fredholm feedback determinant of $(T,S)$ is
	\[
	\Delta_{T,S}(q)
	=
	\DetF\bigl(I_C-qST\bigr).
	\]
	With the heat parametrisation $q=e^{-t}$, where $t>0$, we write
	\[
	\Delta_{T,S}(t)
	=
	\DetF\bigl(I_C-e^{-t}ST\bigr).
	\]
\end{definition}

Since $ST$ is trace class, the map
\[
q\longmapsto\Delta_{T,S}(q)
\]
is entire and satisfies
\[
\Delta_{T,S}(0)=1.
\]
The determinant is therefore defined even at values of $q$ for which the
feedback operator is not invertible. At such values, it vanishes.

\begin{theorem}
	\label{thm:determinant-singularities}
	Under Assumption~\ref{ass:hilbert-schmidt-feedback}, the following
	statements hold.
	\begin{enumerate}
		\item For every $q\in\C$,
		\[
		\Delta_{T,S}(q)=0
		\]
		if and only if $I_C-qST$ is not invertible.
		
		\item If $(\lambda_j)_{j\geq 1}$ is the sequence of nonzero eigenvalues of
		$ST$, repeated according to algebraic multiplicity, then
		\[
		\Delta_{T,S}(q)
		=
		\prod_{j\geq 1}(1-q\lambda_j).
		\]
		
		\item One has the Sylvester identity
		\[
		\DetF\bigl(I_C-qST\bigr)
		=
		\DetF\bigl(I_R-qTS\bigr).
		\]
	\end{enumerate}
\end{theorem}

\begin{proof}
	The first two assertions are standard properties of the Fredholm determinant
	of a trace-class perturbation of the identity. For the third assertion, the
	nonzero eigenvalues of $ST$ and $TS$ coincide with the same algebraic
	multiplicities. Their Fredholm products are therefore equal.
\end{proof}

By Theorem~\ref{thm:transversality-criterion}, the zeros of
$\Delta_{T,S}(q)$ are exactly the values of $q$ for which the graph
subspaces $\Gamma(qT)$ and $\Gamma(S)$ fail to be complementary. Thus, the
determinant is a scalar detector of singular feedback.

\subsection{Trace expansion and closed loop executions}
\label{subsec:trace-expansion}

For sufficiently small $q$, the logarithm of the feedback determinant has
the convergent expansion
\[
\log\Delta_{T,S}(q)
=
-\sum_{m=1}^{\infty}
\frac{q^m}{m}
\Tr\bigl((ST)^m\bigr).
\]
More precisely, this formula is valid whenever
\[
|q|\,\lVert ST\rVert<1,
\]
and it extends by analytic continuation on every simply connected component
of the set on which $\Delta_{T,S}$ does not vanish.

The term
\[
\Tr\bigl((ST)^m\bigr)
\]
is the trace of an execution that traverses the feedback interface $m$ times
and is then closed at the continuation interface. The factor $1/m$ removes
the choice of a marked starting point on the resulting cyclic execution. In
this sense, $\log\Delta_{T,S}$ is a generating function for closed connected
feedback loops, whereas $\Delta_{T,S}$ contains arbitrary finite
collections of such loops.

\begin{proposition}
	\label{prop:logarithmic-derivative}
	Let $\Omega$ be an open subset of $\C$ on which $I_C-qST$ is invertible.
	Then
	\[
	\frac{d}{dq}\log\Delta_{T,S}(q)
	=
	-\Tr\left(
	(I_C-qST)^{-1}ST
	\right)
	\]
	for every $q\in\Omega$. Equivalently, for $q=e^{-t}$,
	\[
	\frac{d}{dt}\log\Delta_{T,S}(t)
	=
	e^{-t}
	\Tr\left(
	(I_C-e^{-t}ST)^{-1}ST
	\right).
	\]
\end{proposition}

\begin{proof}
	For a differentiable trace-class family $A(q)$, the Fredholm determinant
	satisfies
	\[
	\frac{d}{dq}
	\log\DetF\bigl(I_C+A(q)\bigr)
	=
	\Tr\left(
	\bigl(I_C+A(q)\bigr)^{-1}A'(q)
	\right)
	\]
	whenever $I_C+A(q)$ is invertible. Apply this identity to
	\[
	A(q)=-qST.
	\]
	The formula in the variable $t$ follows from the chain rule.
\end{proof}

\subsection{Formal determinants of graded return series}
\label{subsec:formal-graded-determinants}

The graph semantics generally produces a positively graded return series
rather than a single return operator. Let
\[
K(q)
=
\sum_{n=1}^{\infty}q^nK_n,
\qquad
K_n\in\Sone(C).
\]
We first regard this as a formal series in $\Sone(C)[[q]]$.

\begin{definition}
	\label{def:formal-feedback-determinant}
	The formal feedback determinant associated with $K(q)$ is
	\[
	\Delta_{K}^{\mathrm{form}}(q)
	=
	\exp\left(
	-\sum_{m=1}^{\infty}
	\frac{1}{m}
	\Tr\bigl(K(q)^m\bigr)
	\right).
	\]
\end{definition}

The expression is well defined in $\C[[q]]$. Indeed, $K(q)$ has strictly
positive valuation, so the coefficient of $q^N$ receives contributions only
from terms with $m\leq N$. Each such coefficient is a finite sum of traces
of ordered products
\[
\Tr\bigl(K_{n_1}\cdots K_{n_m}\bigr),
\qquad
n_1+\cdots+n_m=N.
\]
This is the determinant-level analogue of the local finiteness used for the
graded execution resolvent in
Proposition~\ref{prop:graded-formal-resolvent}. It is also an instance of the
general formal-series mechanism described in
\cite{Magnot2025FormalSeries}.

\begin{proposition}
	\label{prop:formal-analytic-determinant}
	Assume that there exists $r>0$ such that
	\[
	\sum_{n=1}^{\infty}
	|q|^n\lVert K_n\rVert_1
	<
	\infty
	\]
	for every $q\in\C$ with $|q|<r$. Then $K(q)$ converges in trace norm on the
	disk $|q|<r$, and
	\[
	\Delta_{K}^{\mathrm{form}}(q)
	=
	\DetF\bigl(I_C-K(q)\bigr)
	\]
	as analytic functions on that disk.
\end{proposition}

\begin{proof}
	The stated summability condition gives trace-norm convergence of $K(q)$ and
	locally uniform convergence on compact subdisks. For $q$ sufficiently close
	to zero, one has
	\[
	\lVert K(q)\rVert<1,
	\]
	and the standard logarithmic expansion of the Fredholm determinant gives
	\[
	\DetF\bigl(I_C-K(q)\bigr)
	=
	\exp\left(
	-\sum_{m=1}^{\infty}
	\frac{1}{m}
	\Tr\bigl(K(q)^m\bigr)
	\right).
	\]
	Both sides are analytic on $|q|<r$, so the identity theorem proves the
	result throughout the disk.
\end{proof}

If $K_n$ is the total weight of elementary return cycles of degree $n$, the
coefficient of $q^N$ in
\[
\log\Delta_{K}^{\mathrm{form}}(q)
\]
is a finite sum over closed cyclic concatenations of total degree $N$. Thus,
the formal determinant retains the grading of the execution semantics while
producing a scalar invariant.

\subsection{Covariance and semantic invariance}
\label{subsec:determinant-covariance}

Let
\[
U:C\longrightarrow C'
\qquad\text{and}\qquad
V:R\longrightarrow R'
\]
be bounded invertible operators, and define
\[
T'=VTU^{-1},
\qquad
S'=USV^{-1}.
\]
As shown in Proposition~\ref{prop:feedback-covariance}, one has
\[
S'T'=U(ST)U^{-1}.
\]

\begin{proposition}
	\label{prop:determinant-covariance}
	Under Assumption~\ref{ass:hilbert-schmidt-feedback},
	\[
	\Delta_{T',S'}(q)
	=
	\Delta_{T,S}(q)
	\]
	for every $q\in\C$.
\end{proposition}

\begin{proof}
	The trace-class operators $ST$ and $S'T'$ are similar. The Fredholm
	determinant is invariant under bounded similarity transformations, and hence
	\[
	\DetF\bigl(I_{C'}-qS'T'\bigr)
	=
	\DetF\bigl(I_C-qST\bigr).
	\]
\end{proof}

Consequently, the determinant does not depend on bases or on equivalent
choices of coordinates at the continuation and result interfaces. More
generally, suppose that two closed program presentations induce feedback
configurations whose return operators are intertwined by a bounded
isomorphism. Their Fredholm feedback determinants coincide. This is the
precise invariance statement used here. A stronger assertion for arbitrary
program equivalences would require the extraction of the Hilbert-space
feedback configuration to be canonical with respect to the chosen notion of
denotational equivalence.

\subsection{Removal of the damping parameter}
\label{subsec:determinant-limit}

The value $q=1$ removes the exponential damping of repeated feedback. If
$ST$ is trace class, no renormalisation is needed to define the endpoint
value:
\[
\Delta_{T,S}(1)
=
\DetF(I_C-ST).
\]
It may vanish, in which case the unregularised feedback equation is singular.

\begin{proposition}
	\label{prop:determinant-endpoint}
	Under Assumption~\ref{ass:hilbert-schmidt-feedback},
	\[
	\lim_{q\to 1^{-}}
	\Delta_{T,S}(q)
	=
	\DetF(I_C-ST).
	\]
	Equivalently,
	\[
	\lim_{t\to 0^{+}}
	\Delta_{T,S}(t)
	=
	\DetF(I_C-ST).
	\]
\end{proposition}

\begin{proof}
	One has
	\[
	qST\longrightarrow ST
	\]
	in trace norm as $q\to 1^{-}$. The Fredholm determinant is continuous with
	respect to the trace norm, which proves the first limit. The second follows
	from $q=e^{-t}$.
\end{proof}

For a genuinely graded family $K(t)$, the same conclusion holds whenever
there exists $K_0\in\Sone(C)$ such that
\[
K(t)\longrightarrow K_0
\]
in trace norm as $t\to 0^{+}$. Then
\[
\DetF\bigl(I_C-K(t)\bigr)
\longrightarrow
\DetF(I_C-K_0).
\]
If trace-norm convergence fails, the determinant may diverge or vanish with a
nontrivial asymptotic rate. A renormalised endpoint can then be introduced
only after an asymptotic expansion has been established.

\begin{definition}
	\label{def:renormalised-feedback-determinant}
	Assume that $\Delta(t)\neq 0$ for sufficiently small $t>0$ and that, for a
	continuous choice of the logarithm,
	\[
	\log\Delta(t)
	=
	C_{\mathrm{div}}(t)+c_0+o(1)
	\]
	as $t\to 0^{+}$, where $C_{\mathrm{div}}(t)$ is a specified divergent
	counterterm. The renormalised feedback determinant is
	\[
	\Delta_{\mathrm{ren}}
	=
	\exp(c_0)
	=
	\lim_{t\to 0^{+}}
	\exp\bigl(-C_{\mathrm{div}}(t)\bigr)\Delta(t),
	\]
	provided that the limit exists.
\end{definition}

Definition~\ref{def:renormalised-feedback-determinant} is conditional: the
existence and uniqueness of the counterterm are not consequences of the dcpo
semantics alone. In the present paper, renormalisation will be used only in
examples for which the small-$t$ asymptotics can be computed explicitly.

\subsection{A semantic Fredholm invariant}
\label{subsec:semantic-fredholm-invariant}

Suppose that a closed program $P$ admits a linear feedback presentation
\[
(C_P,R_P,T_P,S_P)
\]
satisfying Assumption~\ref{ass:hilbert-schmidt-feedback}. Define
\[
\Delta_P(q)
=
\DetF\bigl(I_{C_P}-qS_PT_P\bigr).
\]
Then:
\begin{enumerate}
	\item $\Delta_P(q)$ is entire in $q$;
	
	\item its zeros are exactly the parameters for which the regularised
	feedback closure is singular;
	
	\item its logarithm generates traces of closed loop executions;
	
	\item it is invariant under bounded isomorphisms of feedback presentations;
	
	\item its value at $q=1$ is the unregularised Fredholm determinant whenever
	the return operator remains trace class.
\end{enumerate}

The determinant is therefore a derived invariant of the linear feedback
semantics. It supplements, but does not replace, the order-theoretic
denotation constructed in Sections~\ref{sec:graph-semantics}
and~\ref{sec:abel}. The examples in the next section illustrate both the
Abel reconstruction and the Fredholm invariant in concrete recursive
programs.

\section{Examples}
\label{sec:examples}

This section illustrates the constructions of the preceding sections in a
sequence of examples of increasing analytic complexity. The first two examples
come directly from recursive hybrid quantum programs. The remaining examples
focus on the linear feedback sector and show how ordinary and regularised
Fredholm determinants arise from execution resolvents.

\subsection{Repeat-until-success as a scalar feedback loop}
\label{subsec:example-rus}

Consider again the repeat-until-success computation of
Sections~\ref{subsec:rus}, \ref{subsec:rus-graphs}, and
\ref{subsec:rus-abel}. A failed trial has quantum weight
\[
(1-p)\id_{\mathcal A},
\]
whereas a successful trial has quantum weight
\[
p\mathcal U.
\]
A terminating execution with exactly $n$ failures followed by one success has
degree $n+1$ and weight
\[
p(1-p)^n\mathcal U.
\]
Consequently, its degree-weighted graph denotation is
\[
Z_{\mathrm{RUS}}(q)
=
\sum_{n=0}^{\infty}
q^{n+1}p(1-p)^n\mathcal U.
\]
For
\[
|q|<(1-p)^{-1},
\]
the series converges in operator norm and gives
\[
Z_{\mathrm{RUS}}(q)
=
\frac{pq}{1-(1-p)q}\mathcal U.
\]
The limit at $q=1$ is therefore
\[
\lim_{q\to 1^{-}}Z_{\mathrm{RUS}}(q)
=
\mathcal U,
\]
in agreement with the least-fixed-point semantics.

The same calculation can be written as a feedback resolvent. Let the scalar
return factor be
\[
K=1-p.
\]
The success edge contributes $pq\mathcal U$, and hence
\[
Z_{\mathrm{RUS}}(q)
=
pq(1-qK)^{-1}\mathcal U.
\]
The associated feedback determinant is
\[
\Delta_{\mathrm{RUS}}(q)
=
1-(1-p)q.
\]
Its only zero is
\[
q=(1-p)^{-1}>1.
\]
Thus, the feedback operator is nonsingular throughout the physical interval
$0\leq q\leq 1$, and the damping parameter can be removed without
renormalisation.

With $q=e^{-t}$, one obtains
\[
Z_{\mathrm{RUS}}(t)
=
\frac{pe^{-t}}{1-(1-p)e^{-t}}\mathcal U
\]
and
\[
\Delta_{\mathrm{RUS}}(t)
=
1-(1-p)e^{-t}.
\]
Both quantities have regular limits as $t\to 0^{+}$.

\subsection{A measurement-controlled loop with two terminal outcomes}
\label{subsec:example-two-outcomes}

Let $\mathcal A$ be a von Neumann algebra, and let
\[
A_0,A_1,R\in\Wstar(\mathcal A,\mathcal A)
\]
satisfy
\[
A_0(1_{\mathcal A})
+
A_1(1_{\mathcal A})
+
R(1_{\mathcal A})
\leq
1_{\mathcal A}.
\]
The maps $A_0$ and $A_1$ describe two terminating measurement outcomes, while
$R$ describes a retry outcome. The corresponding recursion functional on
$\QO(\mathcal A,\mathcal A,\{0,1\})$ is
\[
F(\xi)
=
\delta(0,A_0)
+
\delta(1,A_1)
+
R\triangleright\xi.
\]
Its $n$-th Kleene approximant is
\[
F^n(\bottom)
=
\sum_{j=0}^{n-1}
\left(
\delta\bigl(0,R^j\circ A_0\bigr)
+
\delta\bigl(1,R^j\circ A_1\bigr)
\right).
\]
Indeed, the term indexed by $j$ represents $j$ successive retries followed by
one of the two terminal outcomes.

The exact-depth graph contribution is
\[
D_{j+1}
=
\delta\bigl(0,R^j\circ A_0\bigr)
+
\delta\bigl(1,R^j\circ A_1\bigr),
\qquad j\geq 0.
\]
The Abel-regularised graph denotation is therefore
\[
Z(q)
=
\sum_{j=0}^{\infty}q^{j+1}
\left(
\delta\bigl(0,R^j\circ A_0\bigr)
+
\delta\bigl(1,R^j\circ A_1\bigr)
\right).
\]
This series is well defined in the quantum orchestra dcpo for every
$0<q\leq 1$. The order-theoretic Abel reconstruction theorem gives
\[
\sup_{0<q<1}Z(q)
=
\fix(F).
\]

Suppose, in addition, that the channels are represented as bounded operators
on a Banach space of observables and that
\[
r(R)<1,
\]
where $r(R)$ is the spectral radius of $R$. Then the series converges in
operator norm at $q=1$, and the two terminal components of the denotation are
\[
\sum_{j=0}^{\infty}R^j\circ A_0
=
(I-R)^{-1}A_0
\]
and
\[
\sum_{j=0}^{\infty}R^j\circ A_1
=
(I-R)^{-1}A_1.
\]
For $0<q<1$, the corresponding regularised expressions are
\[
q(I-qR)^{-1}\circ A_0
\]
and
\[
q(I-qR)^{-1}\circ A_1.
\]
Thus, the same execution resolvent simultaneously controls all terminal
outcomes of the measurement-dependent loop.

\subsection{A finite-dimensional feedback determinant}
\label{subsec:example-finite-dimensional}

Let
\[
C=R=\C^d,
\]
and let $T,S\in\mathcal L(\C^d)$ be such that
\[
ST
=
\operatorname{diag}(\lambda_1,\ldots,\lambda_d).
\]
The regularised execution resolvent is
\[
(I_C-qST)^{-1}
=
\operatorname{diag}
\left(
\frac{1}{1-q\lambda_1},
\ldots,
\frac{1}{1-q\lambda_d}
\right),
\]
whenever
\[
q\lambda_j\neq 1
\]
for every $j$. The feedback determinant is
\[
\Delta_{T,S}(q)
=
\prod_{j=1}^{d}(1-q\lambda_j).
\]
Hence the singular parameters are precisely
\[
q=\lambda_j^{-1}
\]
for the nonzero eigenvalues of $ST$.

For example, in dimension two,
\[
ST
=
\begin{pmatrix}
	\alpha & 0\\
	0 & \beta
\end{pmatrix}
\]
gives
\[
\Delta_{T,S}(q)
=
(1-q\alpha)(1-q\beta).
\]
The logarithmic derivative is
\[
\frac{d}{dq}\log\Delta_{T,S}(q)
=
-
\frac{\alpha}{1-q\alpha}
-
\frac{\beta}{1-q\beta},
\]
which is the trace of the regularised loop resolvent with the sign prescribed
by Proposition~\ref{prop:logarithmic-derivative}.

This example also makes the cross-ratio interpretation explicit. The graph
subspaces $\Gamma(qT)$ and $\Gamma(S)$ are complementary exactly when
\[
(1-q\alpha)(1-q\beta)\neq 0.
\]
Thus, the determinant vanishes exactly when the two graph decompositions cease
to be transverse.

\subsection{A trace-class diagonal return operator}
\label{subsec:example-trace-class-diagonal}

Let
\[
C=R=\ell^2(\N)
\]
with canonical orthonormal basis $(e_n)_{n\geq 1}$, and fix a parameter
\[
0<a<1.
\]
Define Hilbert--Schmidt operators $T$ and $S$ by
\[
Te_n=a^{n/2}e_n
\]
and
\[
Se_n=a^{n/2}e_n.
\]
Then
\[
STe_n=a^ne_n,
\]
and $ST$ is trace class because
\[
\Tr(ST)
=
\sum_{n=1}^{\infty}a^n
=
\frac{a}{1-a}.
\]
The feedback determinant is the convergent $q$-product
\[
\Delta_a(q)
=
\DetF(I_C-qST)
=
\prod_{n=1}^{\infty}(1-qa^n).
\]
It is defined for every $q\in\C$ and is nonzero at $q=1$ because
\[
\sum_{n=1}^{\infty}a^n<\infty
\]
and $a^n\neq 1$.

For $|q|<a^{-1}$, its logarithm can be computed from the trace expansion:
\[
\begin{aligned}
	\log\Delta_a(q)
	&=
	-
	\sum_{m=1}^{\infty}
	\frac{q^m}{m}\Tr\bigl((ST)^m\bigr)
	\\
	&=
	-
	\sum_{m=1}^{\infty}
	\frac{q^m}{m}
	\sum_{n=1}^{\infty}a^{nm}
	\\
	&=
	-
	\sum_{m=1}^{\infty}
	\frac{q^m}{m}
	\frac{a^m}{1-a^m}.
\end{aligned}
\]
The endpoint value is
\[
\Delta_a(1)
=
\prod_{n=1}^{\infty}(1-a^n).
\]
This example shows that an infinite-dimensional feedback determinant may be a
genuine convergent $q$-series determinant without requiring any subtraction
at $q=1$.

\subsection{A heat-regularised determinant with a singular endpoint}
\label{subsec:example-heat-determinant}

We finally consider a model in which the regularised return operator is trace
class for every $t>0$, but no trace-class operator remains at $t=0$. Let
\[
C=\ell^2(\N)
\]
and let $N$ be the number operator
\[
Ne_n=ne_n,
\qquad n\geq 1.
\]
For $t>0$, define
\[
K(t)=e^{-tN}.
\]
Then
\[
K(t)e_n=e^{-tn}e_n
\]
and
\[
\Tr(K(t))
=
\sum_{n=1}^{\infty}e^{-tn}
=
\frac{e^{-t}}{1-e^{-t}}.
\]
Thus, $K(t)$ is trace class for every $t>0$, and its Fredholm determinant is
\[
\Delta_{\mathrm{heat}}(t)
=
\DetF(I_C-K(t))
=
\prod_{n=1}^{\infty}(1-e^{-tn}).
\]
Equivalently, with $q=e^{-t}$,
\[
\Delta_{\mathrm{heat}}(q)
=
\prod_{n=1}^{\infty}(1-q^n).
\]
This is the Euler product associated with the degree spectrum of $N$.

At $t=0$, the operator $K(t)$ converges strongly to the identity, which is not
trace class. The determinant has no nonzero unregularised endpoint; in fact,
\[
\Delta_{\mathrm{heat}}(t)
\longrightarrow
0
\]
as $t\to 0^{+}$. The classical small-$t$ asymptotic of the Euler product is
\cite{Hardy1949}
\[
\Delta_{\mathrm{heat}}(t)
=
\sqrt{\frac{2\pi}{t}}
\exp\left(
-\frac{\pi^2}{6t}
+\frac{t}{24}
\right)
\bigl(1+o(1)\bigr).
\]
Consequently, the renormalised determinant
\[
\Delta_{\mathrm{heat}}^{\mathrm{ren}}
=
\lim_{t\to 0^{+}}
\sqrt{\frac{t}{2\pi}}
\exp\left(
\frac{\pi^2}{6t}
-\frac{t}{24}
\right)
\Delta_{\mathrm{heat}}(t)
\]
exists and satisfies
\[
\Delta_{\mathrm{heat}}^{\mathrm{ren}}=1.
\]

This example clarifies the role of the parametrisation $q=e^{-t}$. The
variable $t$ acts as a heat cutoff on the degree spectrum. For positive $t$,
the cutoff produces a trace-class return operator and a genuine Fredholm
determinant. Removing the cutoff requires a renormalisation because the
trace-class condition is lost at $t=0$. Such a renormalisation is not needed
in the trace-class situation of
Section~\ref{subsec:example-trace-class-diagonal}.

\subsection{Comparison of the examples}
\label{subsec:comparison-examples}

The examples display three distinct behaviours.

\begin{enumerate}
	\item In the repeat-until-success and two-outcome loops, the dcpo semantics is
	reconstructed directly from positive execution contributions. The parameter
	$q=e^{-t}$ regularises the depth of recursion, but the limit at $q=1$ is
	already finite.
	
	\item In the finite-dimensional and trace-class diagonal feedback sectors,
	the Fredholm determinant extends continuously to $q=1$. Its zeros detect
	singular feedback, and no renormalisation is necessary.
	
	\item In the heat-regularised example, every positive value of $t$ yields a
	trace-class determinant, but the endpoint operator is not trace class. A
	nontrivial small-$t$ subtraction is therefore required to extract a finite
	value.
\end{enumerate}

These cases separate the order-theoretic existence of a recursive denotation
from the analytic existence of a Fredholm determinant. The former is supplied
by directed completeness and Scott continuity. The latter requires additional
Schatten-class estimates and may fail at the undamped endpoint even when the
regularised family is well defined for every $t>0$.

\section{Discussion and further work}
\label{sec:discussion}

The constructions developed in this paper separate three levels that are often
combined in the treatment of recursive hybrid quantum programs. The first is
the order-theoretic denotation of recursion by least fixed points in the
quantum orchestra monad. The second is a combinatorial refinement by graded
execution graphs. The third is an optional linear-analytic layer in which a
feedback presentation gives rise to a resolvent, an algebraic cross-ratio,
and, under Schatten-class assumptions, a Fredholm determinant. Keeping these
levels distinct clarifies both the scope of the results and the additional
hypotheses required at each stage.

\subsection{Summary of the semantic construction}
\label{subsec:discussion-summary}

For a finitary hybrid control system, every terminating execution history is
represented by a finite path whose weight is the ordered composition of the
normal completely positive subunital maps attached to its edges. Path
concatenation defines a graded category, and the corresponding formal series
form a complete graded algebra. The degree records execution length and is
additive under concatenation and continuation grafting.

Semantic evaluation sends a terminating graph to the associated Dirac quantum
orchestra and extends to admissible positive graph polynomials. The
compatibility of graph concatenation with channel composition and of graph
grafting with Kleisli composition shows that the graph construction is
compositional. For locally finite execution series, evaluation is defined as
the directed supremum of the evaluations of the finite degree truncations.

The finite-unfolding correspondence of
Theorem~\ref{thm:finite-unfolding-correspondence} identifies these truncations
with the Kleene approximants of the recursive semantic functional. Therefore,
the graph-series semantics is a conservative refinement of the original
quantum orchestra semantics:
\[
\sem{P}_{\mathrm{gr}}
=
\sem{P}.
\]
The refinement retains information that is absent from the final least fixed
point, including the individual terminating histories, their execution
lengths, and their decomposition into successive control-flow steps.

The degree weighting
\[
R_q[\Gamma]
=
q^{|\Gamma|}[\Gamma],
\qquad 0<q<1,
\]
produces an increasing family of regularised denotations. The Abel
reconstruction theorem gives
\[
\sup_{0<q<1}\sem{P}_{q}
=
\sem{P}.
\]
Equivalently, with $q=e^{-t}$,
\[
\lim_{t\to 0^{+}}^{\mathrm{Scott}}\sem{P}_{t}
=
\sem{P}.
\]
This result requires directed completeness, positivity, and Scott continuity,
but no norm convergence.

\subsection{Scope of the graph model}
\label{subsec:scope-graph-model}

The execution graphs used here are finite paths generated by a finitary
control system. This level of generality is sufficient for the recursive
measurement-controlled programs considered in the examples, but it does not
cover every language feature that may occur in a full hybrid quantum
programming language.

First, the construction assumes finite branching at each elementary command.
Countably or continuously valued measurements require an extension from
finite quantum instruments to suitable measurable or domain-theoretic
instruments. In that setting, sums over outgoing edges must be replaced by
operator-valued integration, and local finiteness of the graph series must be
reformulated.

Second, the presentation uses a fixed von Neumann algebra for the quantum
store. Dynamic allocation, deallocation, and changes of quantum type can be
handled only after replacing the single algebra by a family of algebras
indexed by control locations or types. The execution category then becomes a
typed category whose arrows carry channels with varying domains and
codomains. The formal-series construction extends naturally to this setting,
provided that the grading remains locally finite.

Third, the present results establish denotational adequacy of the graph
expansion with respect to the given recursive semantic functional. They do
not constitute a full-abstraction theorem. Such a theorem would require an
operational equivalence for the programming language and a proof that equality
of graph-series denotations coincides with contextual equivalence. Since the
graph series retains more information than the quantum orchestra denotation,
full abstraction may require an explicit quotient identifying operationally
indistinguishable histories.

\subsection{Meaning and limitations of the Abel parameter}
\label{subsec:meaning-abel-parameter}

The parameter $q$ is introduced by the execution grading. It assigns the
factor $q^{n}$ to a history of degree $n$. It should not, without additional
modelling assumptions, be interpreted as physical time, a transition
probability, or a decoherence parameter. The parametrisation
\[
q=e^{-t}
\]
is useful because it turns degree damping into a multiplicative semigroup,
but the variable $t$ is initially only a regulator conjugate to the execution
degree.

Two related Abel expressions occur in the construction. If $D_n$ denotes the
semantic contribution of histories of exact degree $n$, then
\[
\sum_{n=0}^{\infty}q^nD_n
\]
is the directly weighted execution series. If $A_n$ denotes the cumulative
Kleene approximant, then
\[
(1-q)\sum_{n=0}^{\infty}q^nA_n
\]
represents the same regularised denotation. The factor $1-q$ is required
because every exact-depth contribution occurs in all later cumulative
approximants. It is therefore a canonical Abel normalisation rather than an
independently chosen counterterm.

By contrast, the renormalised Fredholm determinants considered in
Section~\ref{sec:fredholm} may require the subtraction of additional
divergent terms as $t\to 0^{+}$. Such counterterms are not determined by the
dcpo semantics. Their existence and invariance must be proved in each
analytic class under consideration. In particular, a regularised determinant
should not be regarded as a semantic invariant unless the choice of
regularisation and subtraction is shown to be canonical under the relevant
program equivalences.

\subsection{Feedback, cross-ratios, and Fredholm determinants}
\label{subsec:discussion-feedback}

The feedback construction of Section~\ref{sec:feedback} applies after a
linear continuation interface $C$, a linear result interface $R$, and maps
\[
T:C\longrightarrow R,
\qquad
S:R\longrightarrow C
\]
have been extracted from the program. The return operator is $ST$, and the
regularised feedback equation is solved by the execution resolvent
\[
(I_C-qST)^{-1}.
\]
Its Neumann expansion enumerates repeated traversals of the feedback
interface. The identity
\[
\CR(C_0,R_0;\Gamma(qT),\Gamma(S))
=
I_C-qST
\]
shows that invertibility of the feedback operator is equivalent to
transversality of two graph subspaces. This is an algebraic observation and
does not require a smooth structure on a Grassmannian.

The Fredholm determinant
\[
\Delta_{T,S}(q)
=
\DetF(I_C-qST)
\]
requires substantially stronger assumptions. The interfaces must be realised
as Hilbert spaces and the return operator must be trace class; a sufficient
condition is that both $T$ and $S$ are Hilbert--Schmidt. Under these
assumptions, the zeros of the determinant detect singular feedback, while the
trace expansion of its logarithm records closed cyclic traversals of the
loop.

This determinant is a derived invariant of a chosen linear feedback
presentation. The paper proves invariance under bounded changes of interface
coordinates. It does not prove that every quantum orchestra admits a
canonical Hilbert-space feedback presentation, nor that two arbitrary
presentations of the same denotation necessarily yield the same determinant.
Establishing such a result would require a functorial linearisation procedure
compatible with the monadic semantics.

The determinant considered here is not asserted to be a tau function of an
integrable hierarchy. No Pl\"ucker relations, Hirota equations, loop-group
factorisation, or Sato Grassmannian flow is constructed. Such structures may
arise in special feedback families carrying an additional integrable-system
symmetry, but they are not consequences of the graph-series semantics alone.

\subsection{Further directions}
\label{subsec:further-directions}

Several extensions of the present framework appear natural.

\paragraph{Typed and higher-dimensional execution diagrams.}
The path model can be replaced by formal series indexed by a graded small
category with varying quantum interfaces. More complicated concurrency or
resource-sensitive constructs may require trees, directed acyclic graphs, or
higher-dimensional diagrams rather than paths. The essential algebraic
condition is that every coefficient of a fixed total degree be determined by
only finitely many decompositions.

\paragraph{Continuous classical outcomes.}
A measurable version of the construction should combine quantum instruments
with valuations or kernels on continuous classical domains. The corresponding
execution expansion would involve iterated operator-valued integrals. An Abel
reconstruction theorem would then require a monotone-convergence principle
compatible with both the measurable and dcpo structures.

\paragraph{Operational adequacy and full abstraction.}
The graph semantics provides a natural bridge between small-step execution
and denotational recursion. A next step is to define an operational reduction
system whose finite runs are precisely the execution graphs and to prove that
the probability and quantum effect of termination agree with graph
evaluation. Contextual equivalence could then be compared with equality in
the graph-series model and in its quantum orchestra quotient.

\paragraph{Canonical feedback extraction.}
The linear feedback sector would become intrinsically semantic if one could
associate functorially with a recursive orchestra a pair of continuation and
result interfaces together with maps $T$ and $S$. Such a construction may be
related to linearisations of continuation semantics, traced monoidal
categories, or geometry-of-interaction models. It should preserve Kleisli
composition and send recursive closure to an execution resolvent.

\paragraph{Determinants and trace formulas.}
For graded trace-class return series, the coefficients of
\[
\log\DetF(I-K(q))
\]
are sums of traces of closed execution cycles. This suggests analogues of
zeta functions and trace formulas for recursive programs. A satisfactory
theory would need to identify the appropriate equivalence relation on cyclic
executions and to determine when the resulting determinant is compositional
under program substitution.

\paragraph{Geometric refinements.}
When the feedback graph subspaces belong to a restricted Grassmannian, the
algebraic cross-ratio may admit a smooth or diffeological refinement. This
could provide geometric information about families of recursive programs and
their singular loci. Such a development requires additional polarization and
operator-ideal hypotheses and is therefore separate from the semantic results
proved here.

\subsection{Conclusion}
\label{subsec:conclusion}

The main outcome of this paper is a graded execution semantics for recursive
hybrid quantum programs that remains compatible with the quantum orchestra
model. Finite execution histories form a formal graph series, semantic
evaluation is compositional, and finite degree truncations coincide with
Kleene approximants. Exponential damping of the graph degree produces an Abel
family whose Scott limit is the ordinary least-fixed-point denotation.

In a supplementary linear sector, recursive feedback is represented by the
resolvent of $I_C-qST$. The same operator is an algebraic cross-ratio of graph
subspaces, and its Fredholm determinant provides a scalar invariant whenever
the return operator is trace class. These analytic constructions enrich the
execution semantics but depend on hypotheses that are not needed for the
existence of the recursive denotation itself.

The resulting framework therefore distinguishes clearly between the general
order-theoretic semantics of recursion, its combinatorial refinement by
execution graphs, and the additional operator-theoretic invariants available
in restricted feedback sectors.

	\vskip 12pt

\paragraph{\bf Data availability statement} No data are available for this work.

\vskip 12pt

\paragraph{\bf Conflict of interest statement} The authors declare no conflict of interest.

\vskip 12pt

\paragraph{\bf Funding} No funding supported this work.

\vskip 12pt

\paragraph{\bf Acknowledgements} 
J.-P.M. thanks the France 2030 framework programme Centre Henri Lebesgue ANR-11-LABX-0020-01
for creating an attractive mathematical environment.

\vskip 12pt

\paragraph{\bf Author's Note on AI Assistance}
Portions of the text were developed with the assistance of a generative language model (OpenAI ChatGPT, based on the GPT-4 architecture). The AI was used to assist with drafting, editing, and standardizing the bibliography format. All mathematical content, structure, and theoretical constructions were provided, verified, and curated by the authors. The authors assume full responsibility for the correctness, originality, and scholarly integrity of the final manuscript.

%\nocite{*}
\bibliographystyle{elsarticle-harv}
\bibliography{graph_series_references}

\end{document}